%% file: main.tex
\documentclass[sigconf, amsbook]{acmart}

\RequirePackage{fix-cm}

\usepackage{subfig}

\usepackage{amsmath,amssymb,amsfonts}
\usepackage{amsthm}
\usepackage{algorithm}
\usepackage[noend]{algpseudocode}
\usepackage{graphicx}
\usepackage{textcomp}
\usepackage{xcolor}
\usepackage[justification=centering]{caption}
\usepackage{setspace}
\usepackage{multirow}
\usepackage{multicol}
\usepackage{enumitem}

\DeclareMathOperator*{\argmax}{\arg\!\max}

\algnewcommand\algorithmicparfor{\textbf{for}}
\algnewcommand\algorithmicpardo{\textbf{do in parallel}}
\algnewcommand\algorithmicforeach{\textbf{for each}}
\algnewcommand{\IfThenElse}[3]{
  \State \algorithmicif\ #1\ \algorithmicthen\ #2\ \algorithmicelse\ #3}

\algrenewtext{ParFor}[1]{\algorithmicparfor\ #1\ \algorithmicpardo}
\algrenewtext{EndParFor}{}
\algrenewtext{ForEach}[1]{\algorithmicforeach\ #1\ \algorithmicfordo}
\algrenewtext{EndForEach}{}
\algdef{SnE}[PARFOR]{ParFor}{EndParFor}[1]{\algorithmicparfor\ #1\ \algorithmicpardo}%
\algdef{SnE}[FOREACH]{ForEach}{EndForEach}[1]{\algorithmicforeach\ #1}%
\algdef{SE}[DOWHILE]{Do}{doWhile}{\algorithmicdo}[1]{\algorithmicwhile\ #1}%
\algdef{SnE}[WHILEWODO]{Whilewodo}{EndWhilewodo}[1]{\algorithmicwhile\ #1 \{\}}%

\makeatletter
\renewcommand{\Function}[2]{%
  \csname ALG@cmd@\ALG@L @Function\endcsname{#1}{#2}%
  \def\jayden@currentfunction{#1}%
}
\newcommand{\funclabel}[1]{%
  \@bsphack
  \protected@write\@auxout{}{%
    \string\newlabel{#1}{{\jayden@currentfunction}{\thepage}}%
  }%
  \@esphack
}
\makeatother

\newtheorem{lemma}{Lemma}
\newtheorem{claim}{Claim}

\newcommand{\beq}{\begin{equation}}
\newcommand{\eeq}{\end{equation}}
\newcommand{\bea}{\begin{eqnarray}}
\newcommand{\eea}{\end{eqnarray}}

\renewcommand\footnotetextcopyrightpermission[1]{} 
\begin{document}

\title{
Planting Trees for scalable and efficient Canonical Hub Labeling
}

\author{
 Kartik Lakhotia, Qing Dong, Rajgopal Kannan, Viktor Prasanna
}
\affiliation{Ming Hsieh Department of Electrical Engineering\\
 University of Southern California}
\email{{klakhoti, qingdong, rajgopak, prasanna}@usc.edu}


\begin{abstract}
Point-to-Point Shortest Distance (PPSD) query is a crucial primitive in graph database applications. Hub labeling algorithms pre-compute a labeling that converts a PPSD query into a list intersection problem, enabling very fast query response.
However, constructing hub 
labeling is computationally challenging. Even state-of-the-art parallel algorithms based on Pruned Landmark Labeling (PLL), are plagued by large label size, violation of given network hierarchy, poor scalability and inability to process large \looseness=-1 weighted graphs.



In this paper, we develop novel parallel shared and distributed memory algorithms for constructing the Canonical Hub Labeling(CHL) that is minimal 
for a given network hierarchy. To the best of our knowledge, none of the existing parallel algorithms guarantee canonical labeling for weighted graphs.
Our key contribution, the 
PLaNT algorithm, scales well beyond the limits of current practice by completely avoiding inter-node communication.
PLaNT also enables the design of a collaborative label partitioning scheme across multiple nodes for 
completely in-memory processing of massive graphs whose labels cannot fit on a single \looseness=-1node.

Compared to the sequential PLL, we empirically demonstrate upto $47.4 \times$ speedup on a $72$ thread shared-memory platform. On a 64-node cluster, PLaNT achieves an average $42 \times$ speedup over single node execution. Finally, we show how our approach demonstrates superior scalability -  we can process $14\times$ larger graphs (in terms of label size) and construct hub labeling orders of magnitude faster compared to state-of-the-art distributed paraPLL algorithm.

\end{abstract}
\maketitle

\input{introduction}

\input{related}
\input{smp}
\input{dmp}
\input{query}
\input{experiment}

\input{conclusion}

\clearpage
\newpage

\bibliographystyle{abbrv}
\bibliography{main}  


\end{document}

%% file: introduction.tex
\section{INTRODUCTION AND BACKGROUND}

Point-to-Point Shortest Distance (PPSD) computation is  one of the most important primitives
encountered in graph databases. It is used for similarity analysis on biological and social networks, context-aware search on knowledge graphs, route navigation on roads etc. These applications generate a very large number of PPSD queries and making online query response fast is absolutely crucial to their \looseness=-1performance.

One way to answer a PPSD query is to run a traversal algorithm such as Dijkstra, Bellman-Ford or Delta-Stepping. However, even state-of-the-art traversal algorithms and implementations \cite{ligra, galois, julienne, gemini, asynch} have query response times in the order of hundreds of milliseconds on large graphs. This is prohibitively slow, especially if the application generates large number of queries. Another solution is to pre-compute and store all pairs shortest paths. While this approach allows queries to be answered in constant time, it incurs quadratic pre-processing time and storage complexity and hence, is not feasible for large \looseness=-1graphs. 
    
Hub-labeling is a popular alternate approach for PPSD computation. It trades off pre-processing costs with query performance, by pre-computing for each vertex, the distance to a `small' subset of vertices known as {\it hubs}. The set of (hub, distance)-tuples $L_v = \{ (h, d(v, h)) \} $ are known as the \textit{hub-labels} of vertex $v$ with $|L_v|$ its {\it label size}. A hub-labeling can correctly answer any PPSD query if it  satisfies the following {\it cover property}: Every connected pair of vertices $u,v$ are covered by a hub vertex $h$ from their shortest path i.e.
there exists an $h\in SP_{u,v}$\footnote{Table \ref{table:notations} lists some frequently used notations in this paper. For ease of description, we consider $G(V, E, W)$ to be weighted and undirected. However, all labeling approaches described here can be easily extended to directed graphs by using \textit{forward} and \textit{backward} labels for each vertex\cite{abrahamCHL}. Our implementation is indeed, compatible with directed graphs.} 
such that $(h, d(u, h)$ and $(h, d(v, h)$ are in the label set of $u$ and $v$, respectively.
Now, a PPSD query for vertices $u$ and $v$ can be answered by finding the common hub $h$ with minimum cumulative distance to $u$ and \looseness=-1$v$. 

Query response time is clearly dependent on average label size.  However, finding the optimum labeling (with minimum average label size) is known to be NP-hard \cite{cohen2hop}. Let $R_{V \rightarrow \mathbb{N}}$ denote a total order on all vertices i.e. a {\it ranking function}, also known as {\it network hierarchy}. Rather than find the optimal labeling, Abraham et al.\cite{abrahamCHL} conceptualize Canonical Hub Labeling (CHL) in which, for a given shortest path $SP_{u,v}$, \textit{only} the highest-ranked hub $h_{m} = \arg\max_{w\in SP_{u,v}} \{R(w)\}$ is added to the labels of $u$ and $v$. CHL satisfies the cover property and is minimal for a given $R$, as removing any label from it results in a violation of the cover property. Intuitively, a good ranking function $R$ will prioritize highly central vertices (such as highways vs residential streets). Such vertices are good candidates for being hubs - a large number of shortest paths in the graph can be covered with a few labels. Therefore, a labeling which is minimal for a good $R$ will be quite efficient overall. \cite{abrahamCHL} develops a sequential polynomial time algorithm for computing CHL where paths are recursively shortcut by vertices in rank order and label sets pulled from high ranked reachable vertices in the modified graph. However, this algorithm incurs large  pre-processing time rendering it infeasible for labeling large \looseness=-1graphs\cite{vldbExperimental, akibaPLL}.

\begin{table}[tbp]
\centering
\caption{Frequently used notations}
\label{table:notations}
\resizebox{\linewidth}{!}{%
\begin{tabular}{|c|c|}
\hline
$G(V, E, W)$      & a weighted undirected graph with vertex set $V$ and edges $E$ \\ \hline
$n, m$         & number of  vertices and edges; $n=|V|, m = |E|$       \\ \hline
$N_v$          & neighboring vertices of $v$                           \\ \hline
$w_{u,v}$      & weight of edge $e=(u,v) \in E$                        \\ \hline
$R$            & ranking function or network hierarchy                 \\ \hline
$SP_{u,v}$     & (vertices in) shortest path(s) between $u$, $v$    (inclusive) \\ \hline
$SPT_v$        & shortest path tree rooted at vertex $v$               \\ \hline
$d(u, v)$      & shortest path distance between vertices $u$ and $v$   \\ \hline
$(h, d(v, h))$ & a hub label for vertex $v$ with $h$ as the hub        \\ \hline
$L_v$          & set of hub labels for vertex $v$                      \\ \hline
$q$            & number of nodes in a multi-node cluster               \\ \hline
\end{tabular}}
\end{table}

Akiba et al.\cite{akibaPLL} propose Pruned Landmark Labeling (PLL)  which is arguably the most efficient sequential algorithm to compute the CHL. PLL iteratively computes Shortest Path Trees (SPTs) from roots selected in decreasing order of rank.  The efficiency of PLL comes from heavy pruning of the SPTs. As shown in fig.\ref{fig:pllTree}, for every vertex $u$ visited in $SPT_v$, PLL initiates a pre-processing {\it Distance-Query} $(v, u, \delta_{v,u})$ to determine if there exists a hub $h$ in both $L_v$ and $L_u$ such that $d(h, u) + d(h, v)\leq \delta_{v,u}$,  where $\delta_{v,u}$ is the distance to $u$ as found in $SPT_v$. Hub-label $(v, \delta_{v,u})$ is added to $L_u$ only if the query is unable to find such a common hub (we say that in such a case, the query returns \textit{false}). Otherwise, further exploration from $u$ is {\it pruned} and $(v, \delta_{v,u})$ is not added to $L_u$. 
(Note: every node is its own hub by default). 
Despite the heavy pruning, PLL is computationally very demanding. Dong et al.\cite{dongPLL} show that PLL takes several days to process large weighted graph datasets - \textit{coPaper} (15M edges) and \textit{Actor} (33M edges). Note that the ranking $R$ also affects the performance of PLL. Intuitively, ranking vertices with high-degree or high betweenness centrality higher should lead to better pruning when processing lower ranked vertices. Optimizing $R$ is of independent interest and not the focus of \cite{akibaPLL} or this study.

    
Parallelizing CHL construction/PLL comes with a  myriad of challenges. 
Most existing 
approaches \cite{parapll, parallelPLLThesis} attempt to 
construct multiple trees concurrently using parallel threads. 
However, such simple parallelization violates network hierarchy and results in higher label sizes.
Many mission critical applications may require a CHL for a specific network hierarchy. Further, larger label sizes directly impact query performance. Parallelizing label construction over multiple nodes in clusters exacerbates these problems because the hubs generated by a node are not immediately visible to other clusters.
Most importantly, the size of labeling can be significantly larger than the graph itself, stressing the available main memory on a single machine. Although a disk-based labeling algorithm has been proposed previously \cite{jiangDisk}, it is substantially slower than PLL and ill suited to process large networks. To the best of our knowledge, none of the existing parallel approaches resolve this \looseness=-1issue.

In this paper, we systematically address the multiple challenging facets of the parallel CHL construction problem. Two key perspectives drive the development of our algorithmic innovations and  optimizations. First, we approach simultaneous construction of multiple SPTs in PLL as an {\it optimistic parallelization} 
that can result in mistakes. 
We develop PLL-inspired shared-memory parallel Global Local Labeling (GLL) and Distributed-memory Global Local Labeling (DGLL) algorithms \looseness=-1that 
\begin{enumerate}
    \item only make mistakes from which they can recover, and 
    \item efficiently correct those mistakes.
\end{enumerate}

Second, we note that mistake correction in DGLL generates huge amount of label traffic, thus limiting its scalability. Therefore, we shift our focus from parallelizing PLL to the {\it primary problem} of parallel CHL construction.  Drawing insights from the fundamental concepts behind canonical labeling, we develop an \textit{embarrassingly parallel} and \textit{communication avoiding} distributed algorithm called \textbf{PLaNT} (Prune Labels and (do) Not (Prune) Trees). Unlike PLL which prunes SPTs using distance queries but inserts labels for all vertices explored, PLaNT does not prune SPTs but inserts labels selectively. PLaNT 
ensures correctness 
and minimality of output hub labels (for a given $R$) generated
from an SPT 
without consulting previously discovered labels, as shown in fig. \ref{fig:example}. This allows labeling to be partitioned across multiple nodes without increase in communication traffic and enables us to simultaneously scale effective parallelism and memory capacity using multiple cluster nodes. By seamlessly transitioning between PLaNT and DGLL, we achieve both computational efficiency and high scalability.

Overall, our contributions can be summarized as follows:
\begin{itemize}[leftmargin=*]
    \item We develop parallel shared-memory and distributed algorithms that output the minimal hub labeling (CHL) for a given graph $G$ and network hierarchy $R$. None of the existing parallel algorithms guarantee the CHL as output. 
    \item We develop a new embarrassingly parallel algorithm for distributed CHL construction, called PLaNT. PLaNT completely avoids inter-node communication to achieve high scalability at the cost of additional computation. We further propose a hybrid algorithm for \textit{efficient} and \textit{scalable} CHL \looseness=-1construction.
    \item Our algorithms use the memory of multiple cluster nodes in a collaborative fashion to enable completely in-memory processing of graphs whose labels cannot fit on the main memory of a single node. To the best of our knowledge, this is the first work to accomplish this task.
    \item We develop different schemes for label data distribution in a cluster to increase query throughput by utilizing parallel processing power of multiple compute nodes. To the best of our knowledge, none of the existing works use multiple machines to store labeling and compute query \looseness=-1response.
\end{itemize}
 We use $12$ real-world datasets to evaluate our algorithms. Label construction using PLaNT is on an average, $42\times$ faster on $64$ nodes compared to single node execution. Further, our distributed implementation is able to process the LiveJournal\cite{liveJournal} graph with $>100$ GB output label size in $<40$ minutes on $64$ \looseness=-1nodes.

\section {Problem Description and Challenges}

In this paper, given a \textit{weighted} graph $G$ with positive edge weights and a ranking function $R$, we develop solutions to the following three objectives:
\begin{itemize}[leftmargin=*]
    \item P1 $\rightarrow$ Efficiently utilize the available parallelism in shared-memory (multicore CPUs) and distributed-memory (multi node clusters) systems to accelerate CHL construction.
    \item P2 $\rightarrow$ Scale (in-memory) processing to large graphs whose main memory requirements cannot be met by a single \looseness=-1machine.
    \item P3 $\rightarrow$ Given a labeling $L=\cup_{v\in V} L_v$, accelerate query response time 
    by using multi-node parallelism in a cluster.
\end{itemize}
The challenges for parallel construction of canonical labels are manifold--foremost being the  dependency of pruning in an SPT on the labels previously generated by all SPTs rooted at higher ranked vertices. This lends the problem its inherently sequential nature, requiring SPTs to be constructed in a successive manner. For a distributed system, this also leads to high label traffic that is required to efficiently prune the trees on each node. Also note that the average label size can be orders of magnitude greater than the average degree of the graph. Hence, even if the main memory of a single node can accommodate the graph, it may not be feasible to store the complete labeling on each \looseness=-1node.\vspace{-2mm} 


\begin{figure*}[htb]
    \centering
    \subfloat[Graph $G$ with labels from $SPT_{v_1}$ (blue tables), rank $R$ and
    distance \& ancestor initialization]{
        \includegraphics[width=0.2\textwidth]{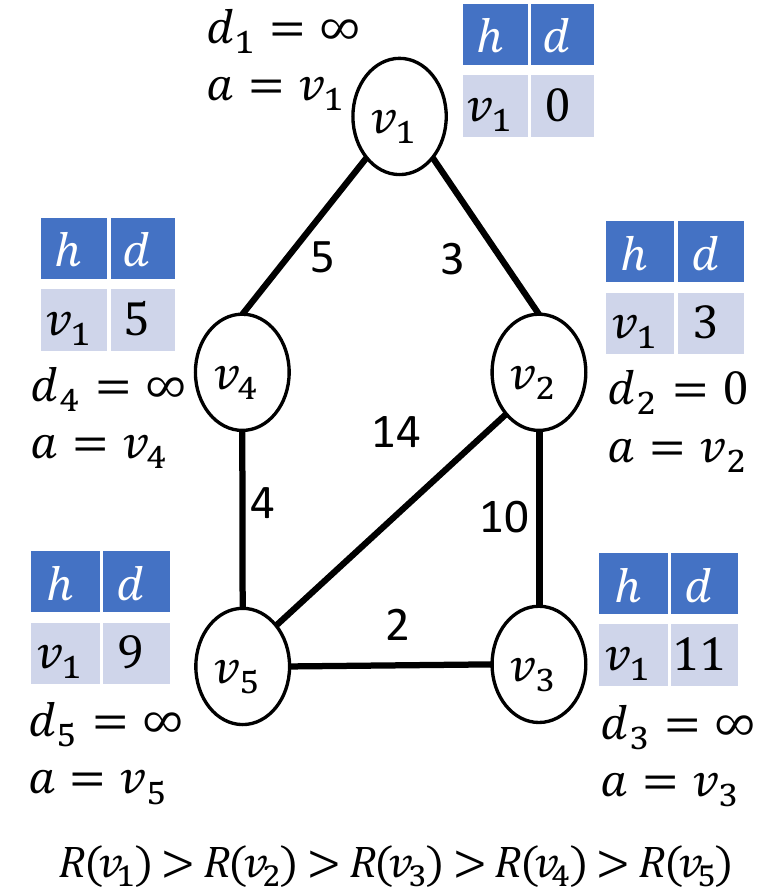}
        \label{fig:initGraph}
    }
    \hfill
    \subfloat[$SPT_{v_2}$ Construction and Label Generation for $G$ in PLL (after $SPT_{v_1}$)]{
        \includegraphics[width=0.75\textwidth]{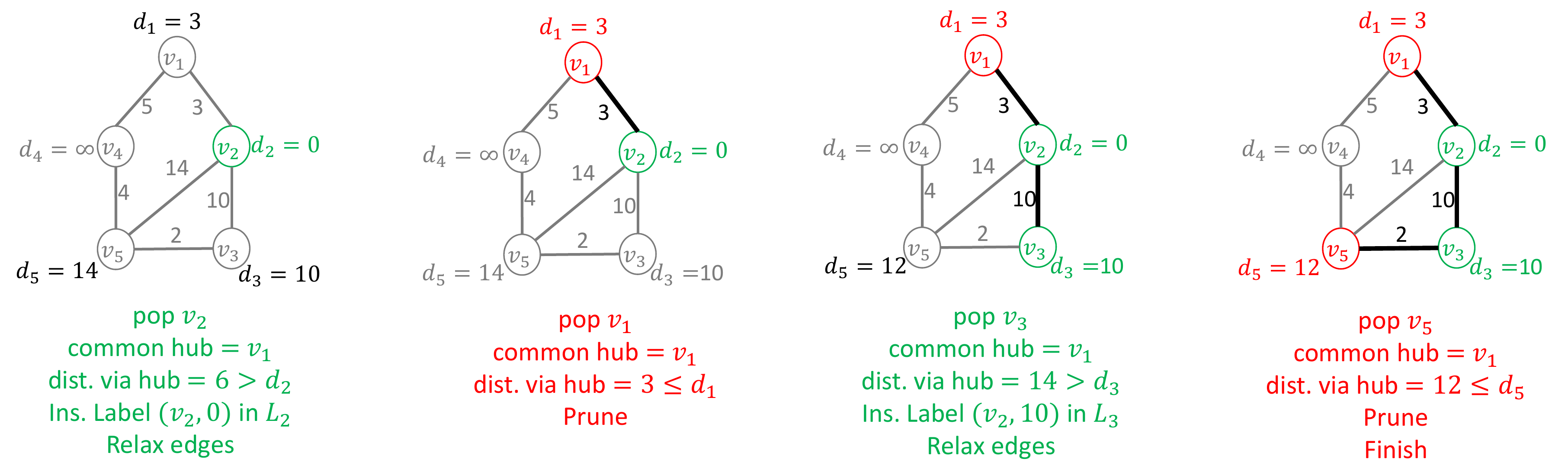}
        \label{fig:pllTree}
    }\\
    \subfloat[$SPT_{v_2}$ Construction and Label Generation for $G$ in PLaNT (after $SPT_{v_1}$)]{
        \includegraphics[width=\textwidth]{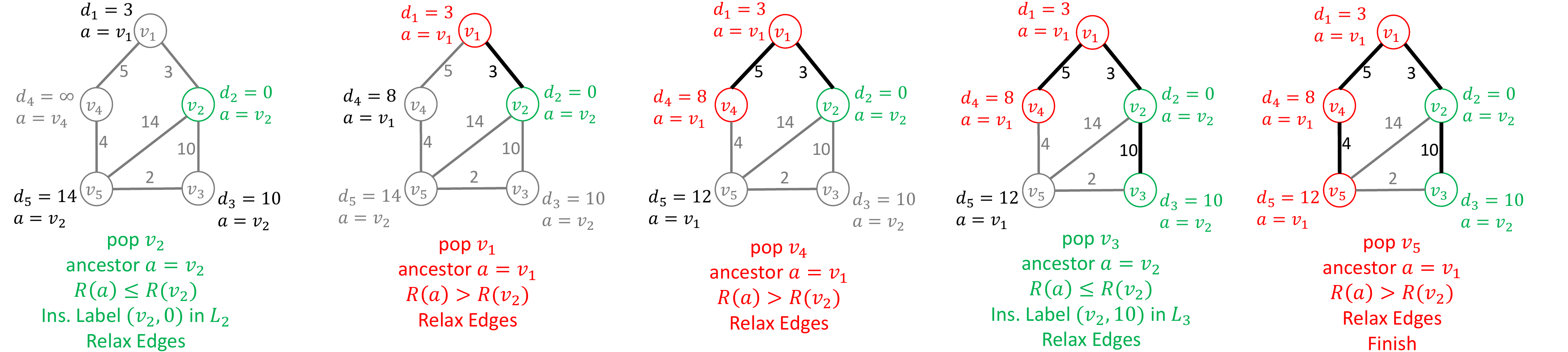}
        \label{fig:plantTree}
    }
    \caption{\looseness=-1 Figures \ref{fig:pllTree} and \ref{fig:plantTree} delineate steps of PLL Dijkstra and PLanT Dijkstra along with the corresponding actions taken
    at each step (Red = label pruned; Green = label generated) for constructing $SPT_{v_2}$.
    For any vertex $v_i$ visited, 
    PLL computes the minimum distance to $v_i$
    through common hubs between $v_2$ and $v_i$. This operation utilizes the previously generated 
    labels from $SPT_{v_1}$ and confirms if $SP_{v_2, v_i}$ is
    already covered by a more important hub.
    Contrarily, PLaNT only uses information intrinsic to $SPT_{v_2}$ 
    by tracking the most important vertex (ancestor $a$) in the shortest path(s) $SP_{v_2, v_i}$.
    PLaNT generates the same (non-redundant) labels as PLL,
    albeit at the cost of additional exploration in $SPT_{v_2}$. \vspace{-2mm}
    }\label{fig:example}
\end{figure*}

%% file: related.tex
\section{Related Work}
Recently, researchers have proposed parallel approaches for hub labeling \cite{parallelPLLThesis, dongPLL, parapll}, many of which are based on PLL. Ferizovic et al.\cite{parallelPLLThesis}
construct a task queue from $R$ such that each thread pops the highest ranked vertex still in the queue and constructs a pruned SPT from that vertex.
They only process unweighted graphs which allows them to use the bit-parallel labeling optimization of \cite{akibaPLL}. Very recently, Li et al.\cite{liSigmod} proposed a highly scalable 
Parallel Shortest distance Labeling (PSL) algorithm.
In a given round $i$, PSL generates 
hub labels with distance $i$ in parallel, replacing the sequential node-order
label dependency of PLL with a distance label dependency. 
PSL is very efficient on and is explicitly designed for 
unweighted small-world networks with low diameter. 
Contrarily, we target a generalized problem of labeling weighted graphs with 
arbitrary diameters, where these approaches are either not applicable or not effective.

Qiu et al.\cite{parapll} argue that existing graph frameworks parallelize a single instance of SSSP and are not suitable for parallelization of PLL. They propose the paraPLL framework that launches concurrent instances of pruned Dijkstra (similar to \cite{parallelPLLThesis}) to process weighted graphs. By using a clever idea of hashing root labels prior to launching an SPT construction, they ensure that despite concurrent tree constructions, the labeling will satisfy the cover property even though it labeling may not be consistent with $R$. While this approach benefits from the order of task assignment and dynamic scheduling, it can lead to significant increase in label size if the number of threads is \looseness=-1large. 

Dong et al. \cite{dongPLL} observe that the sizes of SPTs with high ranked roots is quite large.
They propose a hybrid intra- and inter-tree paralelization scheme utilizing parallel Bellman ford for the large SPTs initially, and concurrent dijkstra instances for small SPTs in the later half of the execution, ensuring average label size close to that of CHL. Their hybrid algorithm works well for scale-free graphs but fails to accelerate high-diameter graphs, such as road networks, due to the high complexity of Bellman Ford. 

paraPLL\cite{parapll} also provides a distributed-memory implementation that statically divides the tasks (root IDs) across multiple nodes. It periodically synchronizes the concurrent processes and exchanges the labels generated so that they can be used by every process for pruning. This generates large amount of label traffic and also introduces a pre-processing vs query performance tradeoff as reducing synchronizations improves labeling time but drastically increases label size. Moreover, paraPLL stores all the labels generated on every node and hence, cannot scale to large graphs despite the cumulative memory of all nodes being enough to store the labels. 

Finally, note that none of the existing parallel approaches \cite{dongPLL, parapll} construct 
the CHL on weighted graphs. All of them generate \textit{redundant labels} (definition \ref{def:redLabel}) and hence do not satisfy the \textit{minimality} property. 



%% file: smp.tex
\section{Shared-memory parallel labeling}\label{sec:smp}
\subsection{Label Construction and Cleaning}
In this section, we discuss LCC - a two-step Label Construction and Cleaning  (LCC) algorithm to generate the CHL for a given graph $G(V, E, W)$ and ordering $R$. LCC utilizes shared-memory parallelism and forms the basis for the other algorithms discussed in this paper.

We first define some labeling properties. Recall that a labeling algorithm produces correct results if it satisfies the cover property.  
\begin{definition}\label{def:redLabel}
A hub label $(h, d(v, h))\in L_v$ is said to be \textit{redundant} if it can be removed from $L_v$ without violating the cover property.
\end{definition}
\begin{definition}\label{def:minLabel}
A labeling $L$ satisfies the {\it minimality property} if it has no redundant labels.
\end{definition}

Let $R$ be any network hierarchy. For any pair of connected vertices $u$ and $v$, let $h_{m} = \arg\max_{w\in SP_{u,v}} \{R(w)\}$. 
\begin{definition}\label{def:respectRank}   
A labeling respects $R$ if 
$(h_{m}, d(u,h_{m}))\in L_u$ and $(h_{m}, d(v,h_{m}))\in L_v$, for all connected vertices $u$, $v$.
\end{definition} 

\begin{lemma}\label{lemma:redLabel}
A hub label $(h, d(v, h))\in L_v$ in a labeling that respects $R$ is \textit{redundant} if $h$
is not the highest ranked vertex in $SP_{v,h}$.
\end{lemma}
\begin{proof}
WLOG, let $w = \arg\max_{u\in SP_{v,h}} \{ R(u) \}$. By assumption, $w \neq h$. By definition, $w \in SP_{v,v'}$, $\forall v' \in S^v_h$, where 
$S^v_h = \{v' | h \in SP_{v',v}$ \}. 
Since the labeling respects $R$, for any $v'$, we must have $(w', d(v, w))\in L_v$ and also $ (w', d(v',w))\in L_{v'}$, where $w'=\arg\max_{u\in SP_{v,v'}} \{ R(u) \}$. Clearly, $R(w')\geq R(w) > R(h)$ which implies that $w'\neq h$. Thus, for every $v'\in S^v_h$, there exists a hub $w'\neq h$ that covers $v$ and $v'$ and $(h, d(v, h))$ can be removed without affecting the cover property.
\end{proof}
    
\begin{lemma}\label{corollary:redLabel}
Given a ranking $R$ and a labeling that respects $R$, a redundant label $(h, d(v, h))\in L_v$ can be detected by a PPSD query between the vertex $v$ and the hub $h$. \end{lemma}
\begin{proof}
Let $w'=\arg\max_{u\in SP_{v,h}} \{ R(u) \}$. By Lemma~\ref{lemma:redLabel}, $w' \neq h$. Further, since the labeling respects $R$, $w'$ must be a hub for $v$ and $h$. Thus a PPSD query between $v$ and $h$ with rank priority will return hub $w'$ and distance $d(v, w') + d(h,w') = d(v, h)$, allowing us to detect redundant label $(h, d(v, h))$ in $L_v$.
 \end{proof}


 Lemmas \ref{lemma:redLabel} and \ref{corollary:redLabel} show that redundant labels (if any) in a labeling can
 be detected if it respects $R$.
Next, we describe our parallel LCC algorithm
and show how it outputs the CHL. Note that the CHL \cite{abrahamCHL} respects $R$ and is minimal.


The main insight underlying LCC is that simultaneous construction of multiple SPTs can be viewed as an {\it optimistic parallelization} of sequential PLL - that allows some `mistakes' (generate 
labels not in CHL)
in the hub labeling. However, only those mistakes shall be allowed that can be corrected to obtain the CHL. LCC addresses two major parallelization  challenges:

    \begin{itemize}[leftmargin=*]
        \item Label Construction $\rightarrow$ Construct in parallel, a labeling that respects $R$.
        \item Label Cleaning $\rightarrow$ Remove all redundant labels in parallel.
    \end{itemize}

\textbf{\textit{Label Construction:}}
To obtain a labeling that respects $R$, LCC's label construction incorporates a crucial element. In addition to Distance-Query pruning, LCC also performs \textit{Rank-Query pruning} (algorithm \ref{alg:pruneDijkstra}--Line \ref{RQ}). Specifically, during construction of $SPT_v$, if a higher ranking vertex $u$ is visited, we 1) prune $SPT_v$ at $u$ and 2) do not insert a label for $v$ into $L_u$ {\it even if the corresponding Distance-Query might have returned false}. Since LCC constructs multiple SPTs in parallel it is possible that the SPT of a higher ranked vertex which should be a hub for $v$ (for example $u$ above) is still in the process of construction and thus the hub list of $v$ is incomplete. Step 2) above guarantees that for any pair of connected vertices $u, v$ with $R(u) > R(v)$, either $u$ is labeled a hub of $v$ or they both share a higher ranked hub. This fact will be crucial in proving the minimal covering property of LCC after its {\it label cleaning} phase. Note that $v$ might get unnecessarily inserted as a hub for some other vertex due to Rank Pruning at $u$. However, as we will show subsequently, such `optimistic' labels can be cleaned (deleted).

The parallel label construction 
in LCC
is shown in algorithm \ref{alg:LCC}. Similar to \cite{parapll, parallelPLLThesis}, each concurrent thread 
selects the most important unselected vertex from $R$ 
(by atomic updates to a global counter), and constructs the corresponding $SPT$ using pruned Dijkstra.
However, unlike previous works, LCC's pruned Dijkstra is also enabled with Rank Queries in addition to
Distance Queries\footnote{Initialization steps only touch array elements that have been modified in the previous run of Dijkstra. We use atomics to avoid race conditions.}. This parallelization strategy exhibits good load balance as all threads are working until the very last SPT and there is no global synchronization barrier where threads may stall. Moreover, pruned Dijkstra is computationally efficient for various network topologies compared to the pruned Bellman Ford of \cite{dongPLL} which performs poorly on large diameter graphs.
    
\begin{algorithm}[htb]
	\caption{Pruned Dijkstra with Rank Queries (\texttt{pruneDijRQ})}
	\label{alg:pruneDijkstra}
	\begin{algorithmic}[1]
	    \Statex{\textbf{Input:} $G(V, E, W)$, $R$, root $h$, current labels $L=\cup_{v\in V}L_v$; \textbf{Output:} hub labels with hub $h$}    
	    \Statex{$\delta_v\rightarrow$ distance to $v$, $Q\rightarrow$ priority queue}
        \State{$LR=\texttt{hash}(L_h)$, $\delta_h=0, \delta_v=\infty\ \forall v\in V\setminus\{h\}$}\Comment{initialize}
        \State{add $(h,0)$ to $Q$}
        \While{$Q$ is not empty}
            \State{pop $(v, \delta_v)$ from $Q$}
            \If {$R(v)>R(h)$} continue \Comment{Rank-Query} \label{RQ}
            \EndIf 
            \If {\texttt{DQ}$(v, h, \delta_v, LR, L_{v})$} continue \Comment{Dist. Query} 
            \EndIf
            \State{$L_v = L_v \cup \{(h, \delta_v)\}$}
            \ForEach{$u\in N_v$}
                \If{$\delta_v+w_{v,u}<\delta_u$} \State{$\delta_u=\delta_v+w_{v,u}$; update $Q$}
                \EndIf
            \EndForEach
            
        \EndWhile
		\Function{\texttt{DQ}}{$v, h, \delta, LR, L_{v}$} \label{l11} 
		    \ForEach{$(h', d(v,h'))\in L_v$}
		        \If{$(h', d(h,h'))\in LR$} \label{l13}
		            \If{$d(v,h')+d(h,h')\leq \delta$} return \texttt{true} \label{l14}
		            \EndIf
		        \EndIf
		    \EndForEach
		    \State{return \texttt{false}}
		\EndFunction		
	\end{algorithmic}
\end{algorithm}

\begin{algorithm}[htb]
	\caption{LCC: Label Construction and Cleaning}
	\label{alg:LCC}
	\begin{algorithmic}[1]
	    \Statex{\textbf{Input:} $G(V, E, W)$, $R$; \textbf{Output:} $L=\cup_{v\in V}L_v$}
	    \Statex{$p\rightarrow$ \# parallel threads, $t_c\rightarrow$ tree count}
	    \Statex{$Q\rightarrow$ queue containing vertices ordered by rank}
        \State{$L_v=\phi\ \forall\ v\in V$}\Comment{initialization}
        
        \ParFor{$t_{id} = 1,2 ... p$} \Comment{LCC-I: Label Construction}
            \While{$Q\neq\ $empty}
                \State{atomically pop highest ranked vertex $h$ from $Q$}
                \State{\texttt{pruneDijRQ}$(G, R, h, L)$}
            \EndWhile
        \EndParFor
        \ParFor{$v\in V$}   \label{l6}
            \State{sort labels in $L_v$ using hub rank}
        \EndParFor
        \ParFor{$v\in V$}\Comment{LCC-II: Label Cleaning} 
            \ForEach{$(h, \delta_{v,h})\in L_v$}
                \If{\texttt{DQ\_Clean}$(v, h, \delta_{v,h}, L_h, L_v, R)$}
                \State{delete $(h, \delta_{v,h})$ from $L_v$}
                \EndIf
            \EndForEach
        \EndParFor
		\Function{\texttt{DQ\_Clean}}{$v, h, \delta, L_h, L_{v}, R$}\Comment{Cleaning Query} \label{l12}
		\State{compute the set $W$ of common hubs in $L_h$ and $L_v$} 
		 \Statex{ \hspace{\algorithmicindent} such that $d(w,v)+d(w,h)\leq\delta\ \forall\ w\in W$}
		\State{find the highest ranked vertex $u$ in $W$}
		\IfThenElse{$(W=empty)$ or $R(u)\leq R(h)$} {return \texttt{false}\\ \hspace{\algorithmicindent}}{return \texttt{true}}\label{l16}
		\EndFunction		
	\end{algorithmic}
\end{algorithm}
    
\begin{claim}\label{claim:lccCover}
The labeling generated by LCC's label construction step (LCC-I) satisfies the cover property and respects $R$.
\end{claim}
\begin{proof}
\looseness=-1 Let $H^P_v$ ($H^S_v$, resp.) denote the set of hub vertices of a vertex $v$ after LCC-I (sequential PLL, resp.). We will show that $H^S_v \subseteq H^P_v$.
Suppose $h \not\in H_v^P$ for some vertex $h$. 
Consider three cases:
 
\noindent {\bf Case 1}: $h \not\in H_v^P$ because a Rank-Query pruned $SPT_h$ at $v$ in LCC-I. Thus we must have $R(v) > R(h)$. Since sequential PLL is also the CHL,  $h \not\in H_v^S$ also.


\noindent {\bf Case 2}: $h \not\in H_v^P$ because a Distance-Query pruned $SPT_h$ at $v$ in LCC-I. This can only happen if LCC found a shorter distance $d(h,v)$ through a hub vertex $h' \in SP_{h,v}$  (alg.~\ref{alg:pruneDijkstra} : lines~\ref{l13}-\ref{l14}). Since LCC with Rank-Querying identified $h'$ as a hub for both $h$ and $v$, we must have $R(h') > R(h) > R(v)$ and thus $h \not\in H_v^S$.

\noindent {\bf Case 3}: $h \not\in H_v^P$ because $v$ was not discovered by $SPT_h$ due to some vertex $u$ being pruned. Similar to Case 2 above, this implies $\exists h' \in SP_{v,h}$ with $R(h') > R(h)$ and therefore $h \not\in H_v^S$.   

Combining these cases, we can say that $H^S_v \subseteq H^P_v$. Since sequential PLL also generates the CHL for $R$, the claim \looseness=-1follows.
\end{proof}

\textbf{\textit{Label Cleaning:}}
Note that LCC creates some extra labels due to the parallel construction of $SPT's$. For example, $v$ might get (incorrectly) inserted as a hub for vertex $w$ if 
the $SPT_{v'}$ for a higher ranked vertex $v'\in SP_{v,w}$ is still under construction and $v'$ has not yet been inserted as a hub for $w$ and $v$. These extra labels are  redundant, since there exists a canonical subset of LCC  
(i.e. $H_V^S$) 
satisfying the cover property, and so do not affect PPSD queries. LCC eliminates redundant labels using the $DQ_Clean$ function alg~\ref{alg:LCC}-lines~\ref{l12}-\ref{l16}\footnote{Instead of computing full set intersection, the actual implementation of \texttt{DQ\_Clean} stops at the first common hub (also the highest ranked) in sorted $L_h$ and $L_v$  which satisfies the condition in line \ref{l16} of algortihm \ref{alg:LCC}.}- For vertex $v$, a label $(h, d(v,h))$ is redundant if a Distance-query $(v, h, d(v,h) )$ returns \texttt{true} with a hub $u$ with $R(u)>R(h)\}$. 

\begin{claim}
The final labeling generated by LCC after the Label Cleaning step (LCC-II) is the CHL.
\end{claim}
\begin{proof}
From claim \ref{claim:lccCover}, we know that the labeling after LCC-I respects $R$. Lemma \ref{corollary:redLabel} implies that LCC-II can be used to detect and remove all redundant labels. Hence, the final labeling generated by LCC is minimal and by definition, the CHL.
\end{proof}
    
\begin{lemma}\label{lemma:complexityLCC}
LCC is work-efficient. It performs \\$\mathcal{O}(wm\log^2{n} + w^2n\log^2{n})$ work, generates $O(wn\log{n})$ hub labels and answers each query in $\mathcal{O}(w\log{n})$ time, where $w$ is the tree-width of $G$. 
\end{lemma}
\begin{proof}
Consider the centroid decomposition\\ $(\chi, T(V_T, E_T))$ of minimum-width tree decomposition of the input graph $G$, where $\chi=\{X_t \subseteq V\ \forall\ t\in V_T\}$ maps vertices in $T$ (bags) to subset of vertices in $G$ \cite{akibaPLL}. Let $R(v)$ be determined by the minimum depth bag $\{b_v\in V_T\ |\ v\in X_{b_v}\}$ i.e. vertices in root bag are ranked highest followed by vertices in children of root and so on. Since we prune using \textit{Rank-Query}, $SPT_v$ will never visit vertices beyond the parent of $b_v$. A bag is mapped to at most $w$ vertices and 
the depth of $T$ is $\mathcal{O}(\log{n})$.
Since the only labels inserted at a vertex are its ancestors in the centroid tree, there are $\mathcal{O}(w \log n)$ labels per vertex.

Each time a label is inserted at a vertex, we evaluate all its neighbors in the distance queue. Thus the total number of distance queue operations is $\mathcal{O}(wm\log{n})$.
Further, distance queries are performed on vertices that cannot be pruned by rank queries. This results in $\mathcal{O}(n \cdot w\log{n}\cdot w\log{n})=\mathcal{O}(w^2n\log^2{n})$ work.
   
    
Label Cleaning step sorts the label sets and executes PPSD queries performing $\mathcal{O}(nw\log{n}\log{w}\log{\log{n}} + w^2n\log^2{n}) = \mathcal{O}(w^2n\log^2{n})$ work. Thus, overall work complexity of LCC is $\mathcal{O}(wm\log^2{n} + w^2n\log^2{n})$ which is the same as the sequential 
algorithm \cite{parapll}, making LCC work-efficient.
    \end{proof}

Note that paraPLL\cite{parapll} generated labeling is not guaranteed to respect $R$ and hence, doing Label Cleaning
after paraPLL may result in a labeling that violates the cover property. However, it can be used to clean the output
of inter-tree parallel algorithm by Dong et al \looseness=-1\cite{dongPLL}.
    
Although LCC is theoretically efficient, in practice, the Label cleaning step adds non-trivial overhead to the execution time. In the next subsection, we describe a Global Local Labeling (GLL) algorithm that drastically reduces the overhead of cleaning. 

    
\subsection{Global Local Labeling (GLL)}\label{sec:gll}

The main goal of the GLL algorithm is to severely restrict the size of label sets used for PPSD queries during Label 
Cleaning. A natural way to accelerate label cleaning is by avoiding futile computations (in \texttt{DQ\_Clean}) over hub labels that were already consulted during label 
construction. However, to achieve notable speedup, 
these pre-consulted labels must be skipped in constant time without actually iterating over all of them.
    
GLL overcomes this challenge by using a novel Global Local Label Table data structure and interleaved cleaning strategy. As opposed to LCC, GLL utilizes multiple synchronizations where the threads switch between label construction and cleaning.
We denote the combination of a Label Construction and corresponding Label Cleaning step as a \textit{superstep}. During label construction, the newly generated labels are pushed to a \textit{Local Label Table} and the volume of labels generated is tracked. Once the number of labels in the local table becomes greater than $\alpha n$, where $\alpha > 1$
is the synchronization threshold, the threads synchronize, sort and clean the labels in local table and commit them to the \textit{Global Label Table}. 

In the next superstep, it is known that all labels in the global table are consulted\footnote{For effective pruning, the Label Construction step uses both global and local table to answer distance queries.} during the label generation. Therefore, the label cleaning only needs to query for redundant labels on the local table, thus dramatically reducing the number of repeated computations in PPSD queries. After a label construction step, the local table holds a total of $\alpha n$ labels. Assuming $\mathcal{O}(\alpha)$ average labels per vertex,  
(we empirically observe that labels are almost uniformly distributed across the vertices except the few highest ranked vertices), each cleaning step should perform on average $\mathcal{O}(n\alpha^2)$ work. The number of cleaning steps is $\mathcal{O}\left(\frac{wn\log{n}}{\alpha n}\right)$ and thus we expect the total complexity of cleaning to be $\mathcal{O}(n\alpha w\log{n})$ in GLL as opposed to $\mathcal{O}(nw^2\log^2{n})$ in LCC. If the constant $\alpha \ll w\log{n}$, cleaning in GLL is more efficient than LCC.
    
Using two tables also drastically reduces locking during pruning queries. Both paraPLL and LCC have to lock label sets before reading because label sets are dynamic arrays that can undergo memory (de)allocation when a label is appended. However, GLL only appends to the local table. Most pruning queries are answered by label sets in the global table that do not need to be locked. 
    




%% file: dmp.tex
\section{Distributed-memory Hub Labeling}\label{sec:dmp}
A distributed algorithm allows the application to scale beyond the levels of parallelism and the main memory offered by a single node. This is particularly useful for hub labeling as it is extremely memory intensive and computationally demanding, rendering off-the-shelf shared-memory systems inapt 
for processing large-scale graphs. However, a distributed-memory system also presents strikingly different challenges than a shared-memory system, in general as well as in the specific context of hub labeling. Therefore, a trivial extension of GLL algorithm is not suitable for a multi-node \looseness=-1cluster. 

Particularly, the labels generated on a node are not readily available to other nodes until nodes synchronize and exchange labels. 
Further, unlike paraPLL, our aim is to harness not just the compute but also the collective memory capability of multiple nodes to construct CHL for large graphs. This mandates that labels be partitioned and distributed across multiple nodes at all times, and severely limits the knowledge of SPTs created by other nodes even after synchronization. This absence of labels dramatically affects the pruning efficiency during label construction, resulting in large number of redundant labels and consequently, huge communication volume that bottlenecks the pre-processing. 

In this section, we will present novel algorithms and optimizations that systematically conquer these challenges. We begin the discussion with a distributed extension of GLL that highlights the basic data distribution and parallelization approach.\vspace{-1mm}

\subsection{Distributed GLL (DGLL)}\label{sec:dgll}
\looseness=-1 The distributed  GLL (DGLL) algorithm 
divides the task queue for SPT creation uniformly among $q$ nodes in a rank circular manner. The set of root vertices assigned to node $i$ is $TQ_i = \{ v \ \ |  R(v)\bmod q = i \}$.
Every node loads the complete graph instance and executes GLL on its alloted task queue\footnote{Every node also stores a copy of complete ranking $R$ for rank queries.}. DGLL has two key optimizations tailored for distributed  implementation: 

\textbf{\textit{1. Label Set Partitioning:}} In DGLL, nodes only store labels generated {\it locally} i.e. all labels at node $i$ are of the form $(h, d(v, h))$, where $h \in TQ_i$. Equivalently, the labels of a vertex $v$ are disjoint and distributed across nodes i.e. $L_v = \cup_i L_{i,v}$. Thus, all the nodes collaborate to provide main memory space for storing the labels and the {\it effective memory scales in proportion to the number of nodes}. This is in stark contrast with paraPLL that stores  $\{\cup_{v\in V}L_v\}$ on every node, rendering effective memory same as that of a single \looseness=-1node.



\textbf{\textit{2. Synchronization and Label Cleaning}}: For every superstep in DGLL, we decide the synchronization point apriori in terms of the number of SPTs to be created. The synchronization point computation is motivated by the label generation behavior of the algorithm. Fig.~\ref{fig:labelVsRank} shows that the number of labels generated by initial SPTs rooted at high rank vertices is very large and it drops exponentially as the rank decreases. To maintain cleaning efficiency with few synchronizations, we increase the number of SPTs constructed in the supersteps by a factor of $\beta$ i.e. if superstep $i$ constructs $x$ SPTs, superstep $i+1$ will construct $\beta\cdot x$ SPTs. This is unlike distributed paraPLL\cite{parapll} where same number of trees are constructed in every superstep.
 
After synchronization, the labels generated in a superstep are broadcasted to all nodes. Each node creates a bitvector containing response of all cleaning queries. The bitvectors are then combined using an all reduce operation to obtain final redundancy \looseness=-1information.

Note that DGLL uses both global and local tables to answer cleaning queries.
Yet, interleaved cleaning is beneficial as it removes redundant labels, thereby reducing query response time for future cleaning steps. While label construction queries only use tables on generator node, cleaning queries use tables on all nodes for every query. The presence of redundant labels can thus, radically slow down cleaning. For some datasets, we empirically observe $>90\%$ redundancy in labels generated in some \looseness=-1supersteps.

\begin{figure}[htbp]
    \centering
 \subfloat{\includegraphics[width=0.5\linewidth]{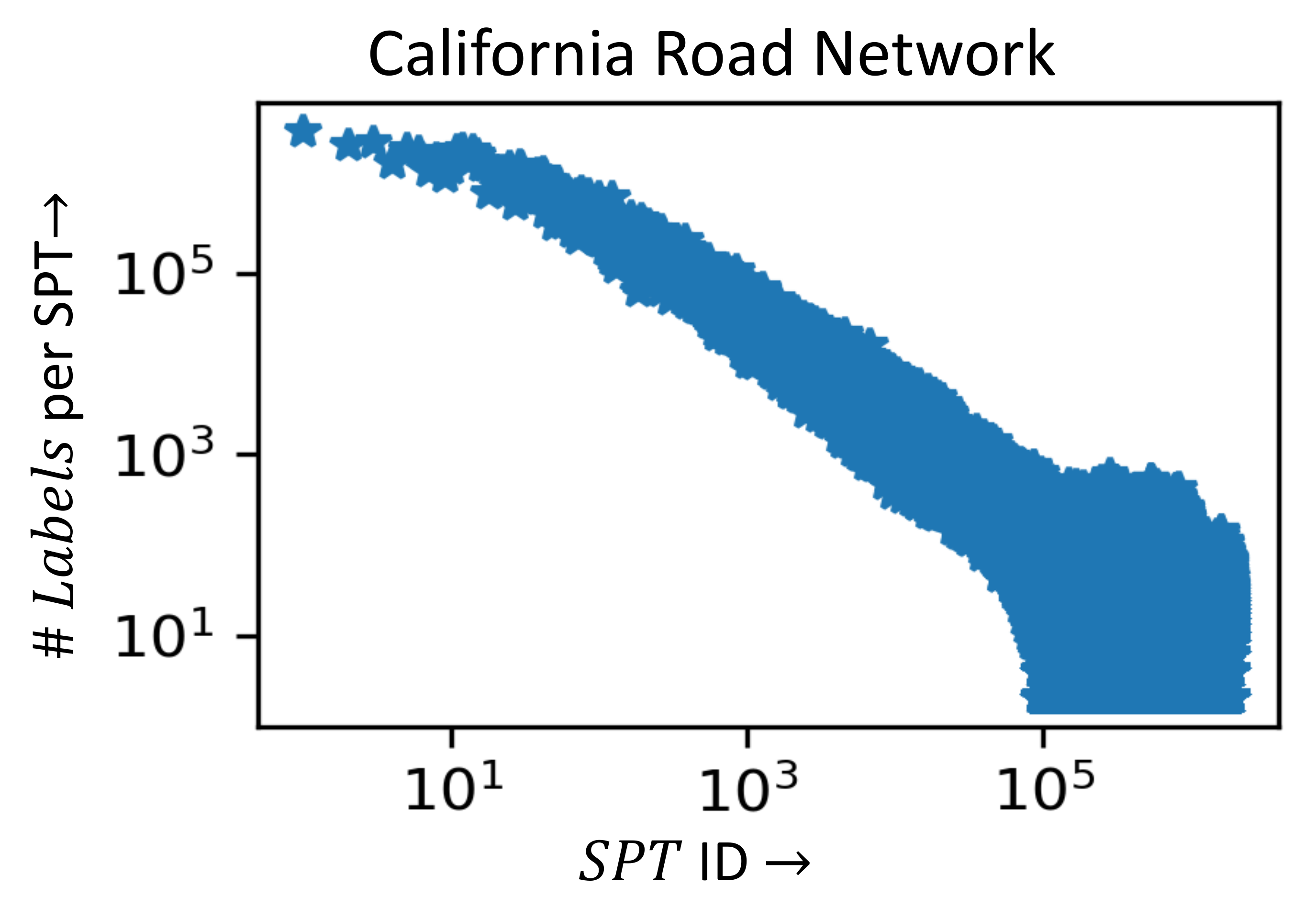}}
 ~~
 \subfloat{\includegraphics[width=0.5\linewidth]{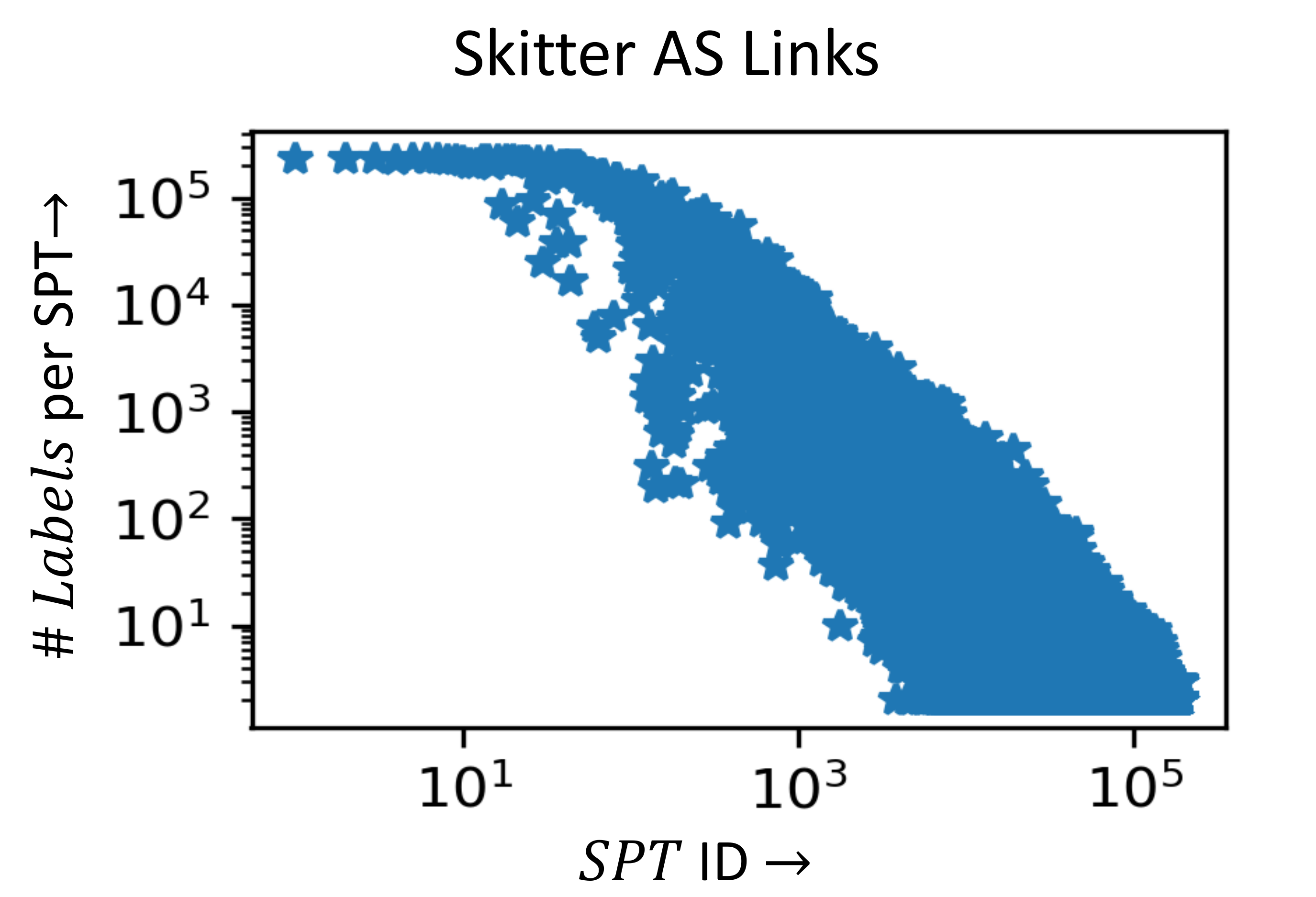}}
 
    \caption{Labels generated by SPTs. ID of $SPT_v$ is $n-R(v)$.\vspace{-2mm}}
    \label{fig:labelVsRank}
\end{figure}

\subsection{Prune Labels and (do) Not (prune) Trees \looseness=-1(PLaNT) }\label{sec:plant}
The redundancy check in DGLL can severely restrict scalability of the algorithm due to huge label broadcast traffic (redundant $+$ non-redundant labels), motivating the need for an algorithm that can avoid redundancy without communicating with other nodes. 

To this purpose, we propose the Prune Labels and (do) Not (prune) Trees (PLaNT) algorithm that 
accepts some loss in pruning efficiency to achieve a dramatic decrease in communication across nodes in order
by outputting completely non-redundant labels without additional label cleaning. We note that the redundancy of a hub label $(h, d(v, h))\in L_v$ is only determined by whether or not $h$ is the highest ranked vertex in $SP_{v,h}$. 
This is the key idea behind PLaNT: When constructing $SPT_h$, if, when resolving distance queries, embedded information about high-ranked vertices on paths can be retrieved, $SPT_h$ will {\it intrinsically} have the requisite information to detect redundancy of $h$ as a \looseness=-1hub. 

Algorithm \ref{alg:plant} (PLaNTDijkstra) depicts the construction of a shortest path tree using PLaNT, which we call PLaNTing trees. Instead of pruning using distance or rank queries, PLaNTDijkstra tracks the {\it most important ancestor} $a[v]$ encountered on the path from $h$ to $v$ by allowing ancestor values to propagate along with  distance values. When $v$ is popped from the distance queue, a label is added to $L_v$ if neither $v$ nor $a[v]$ are ranked higher than the root. Thus, for any shortest path $SP_{h, v}$, only $h_m=\argmax_{w\in SP_{h,v}} \{ R(w) \}$ succeeds in adding itself to the labels of $u$ and $v$, {\it guaranteeing minimality of the labeling while simultaneously respecting $R$}. Figure \ref{fig:plantTree} provides a detailed graphical illustration of label generation using PLaNT and 
shows that it generates the same labeling as the canonical PLL. 

If there are multiple shortest paths from $h$ to $v$, the path with the highest-ranked ancestor is selected.
This is achieved in the following manner: when a vertex $v$ is popped from the dijkstra queue and its edges are relaxed, the
ancestor of a neighbor $u\in N_v$ is allowed to update even if the newly calculated tentative distance to $u$ is \textit{equal} 
to the currently assigned distance to $u$ (line 12 of algorithm \ref{alg:plant}). For example, in fig.\ref{fig:plantTree}, the 
shortest paths to $v_5$,  $P_1 = \{v_2, v_1, v_4, v_5\}$ and $P_2 = \{v_2, v_3, v_5\}$ have the same length and $P_1$ is
selected by setting $a[v_5]=v_1$ because $R(v_1)>R(v_2)$.

Note that PLaNT not only avoids dependency on labels on remote nodes, but it rids SPT construction of \textit{any dependency} on the output of other SPTs, effectively providing an \textbf{embarassingly parallel} solution for CHL construction with $\mathcal{O}(m + n\log{n})$ depth (complexity of a single instance of dijkstra) and $\mathcal{O}(mn + n^2\log{n})$ \looseness=-1work. 
Due to its embarassingly parallel nature, PLaNT does not require SPTs to be constructed in a specific order. However, to enable optimizations discussed later, we follow the same rank determined order in PLaNT as used in DGLL (section \ref{sec:dgll}). 

\textbf{\textit{Early Termination:}} To improve the computational efficiency of PLaNT and prevent it from exploring the full graph $G$ for every SPT, we propose the following simple \textit{early termination} strategy: stop further exploration when the rank of either the ancestor or the vertex itself is higher than root for \textit{all} vertices in dijkstra's distance queue \footnote{Further exploration from such vertices will only result in shortest paths with at least one vertex ranked higher than the root and hence, no labels will be generated.}. Early termination has the potential to dramatically cut down traversal in SPTs with low-ranked roots.

Despite early termination, PLaNTed trees can possibly explore a large part of the graph which PLL would have pruned. Fig.\ref{fig:twVsLabel} shows that in PLaNT, \# vertices explored in an SPT per label generated ($\Psi$) can be $>10^4$. Large value of $\Psi$ implies a lot of exploration overhead that PLL algorithm would have avoided by pruning.

\begin{figure}[htbp]
    \centering
 \subfloat{\includegraphics[width=0.5\linewidth]{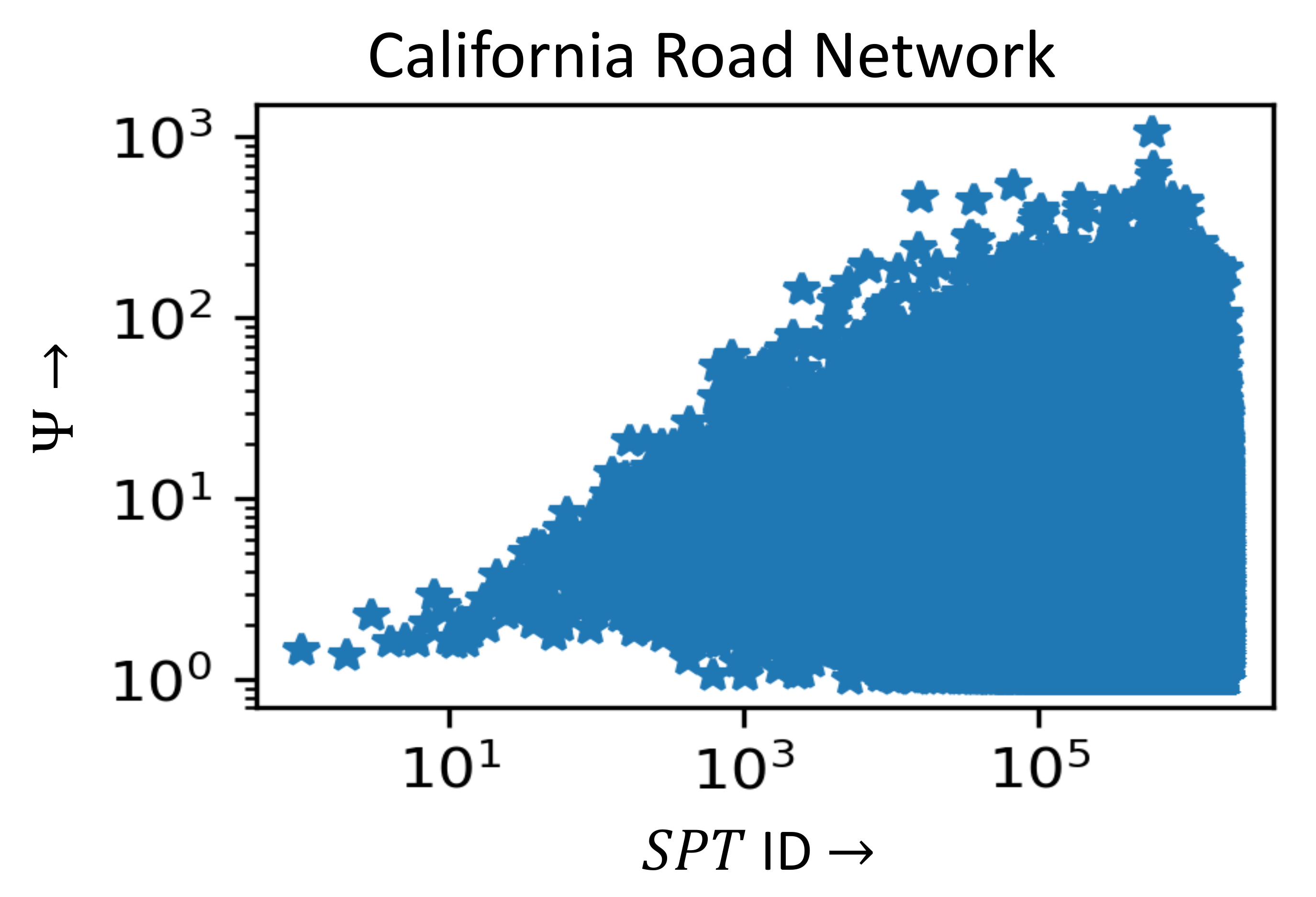}}~~ \subfloat{\includegraphics[width=0.5\linewidth]{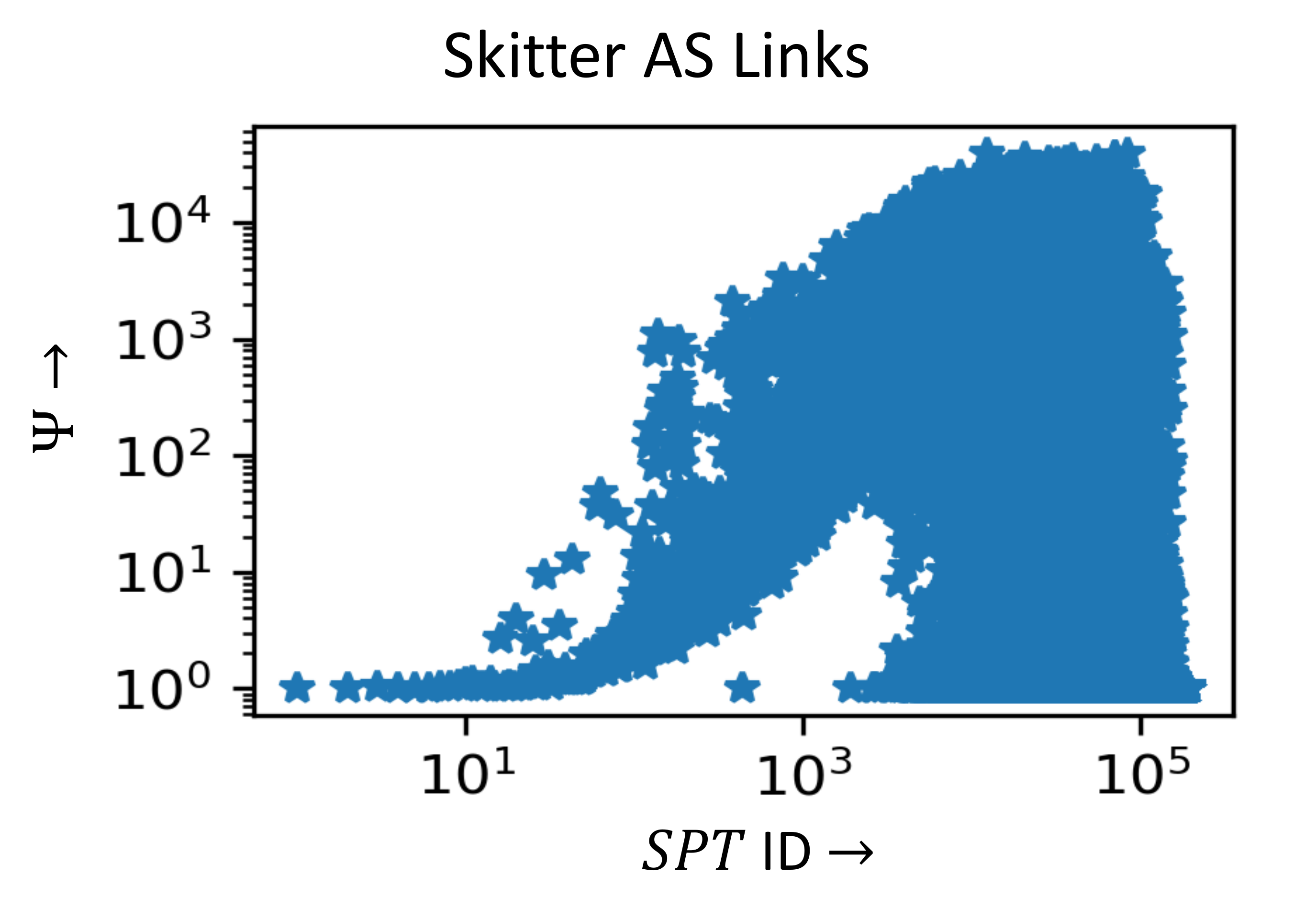}}
    \caption{$\Psi$ (ratio of \# vertices explored in an SPT to the \# labels generated) can be very high for later SPTs.\vspace{-2mm}}
    \label{fig:twVsLabel}
\end{figure}

\begin{algorithm}[]
	\caption{Planting Shortest Path Trees \texttt{(PLaNTDijkstra)}}
	\label{alg:plant}
	\begin{algorithmic}[1]
	    \Statex{\textbf{Input:} $G(V, E, W)$, $R$, root $h$}    
	    \Statex{$\delta_v\rightarrow$ distance to $v$, $a[]\rightarrow$ ancestor array, $Q\rightarrow$ priority queue, $cnt\rightarrow$ number of vertices $v$ with $a[v]=h$}
        \State{$\delta_h=0, a[h]=h$ and $a[v]=v, \delta_v=\infty\ \forall\ v\in V\setminus{h}$}
        \State{add $h$ to $Q$;\ \ $cnt=1$}
        \While{$Q$ is not empty}
            \If{$cnt=0$} exit \Comment{Early Termination}
            \EndIf
            \State{pop $(v,\delta_v)$ from $Q$; compute $nA=\argmax\limits_{x\in\{v, a[v]\}}R(x)$}
            \If{$a[v]=h$}  $cnt=cnt-1$  \EndIf
            \If {$R[nA]>R[h]$} continue
            \EndIf 
            \State{$L_v = L_v \cup \{(h, \delta_v)\}$}
            \ForEach{$u\in N_v$}
                \State{$pA = a[u]$}
                \If{$\delta_v+w_{v,u}<\delta_u$} $a[u]=\argmax\limits_{x\in\{nA, u\}} R(x)$
                \ElsIf{$\delta_v+w_{v,u}=\delta_u$} $a[u]=\argmax\limits_{x\in\{nA, pA\}} R(x)$
                \EndIf
                \If{$a[u]=h$ \text{and} $pA\neq h$} $cnt=cnt+1$
                \ElsIf{$a[u]\neq h$ and $pA=h$} $cnt=cnt-1$
                \EndIf
                \State{$\delta_u=\min(\delta_u, \delta_v+w_{v,u})$;\ \ update $Q$}
                
            \EndForEach
            
        \EndWhile
	\end{algorithmic}
\end{algorithm}


\subsubsection{Hybrid PLaNT + DGLL}
Apart from its embarrassingly parallel nature, an important virtue of PLaNT is its \textit{compatibility} with DGLL. Since PLaNT also constructs SPTs in rank order and generates labels with root as the hub, we can seamlessly transition between PLaNT and DGLL to enjoy the best of both worlds. We propose a \textit{Hybrid} algorithm that initially uses PLaNT and switches to DGLL after certain number of SPTs. The initial SPTs rooted at high ranked vertices generate most of the labels in CHL (fig.\ref{fig:labelVsRank}) and exhibit low $\Psi$ value (fig.\ref{fig:twVsLabel}). By PLaNTing these SPTs, we efficiently parallelize bulk of the computation and avoid communicating a large fraction of the overall label set,
at the cost of little additional exploration in the trees. By doing PLaNT for the initial SPTs, we also avoid a large number of distance queries that PLL or DGLL would have 
done on all the visited vertices in those SPTs. In the later half of execution, $\Psi$ becomes high and very few labels are generated per SPT. The Hybrid algorithm uses DGLL in this half to exploit the heavy pruning and avoid the inefficiencies associated with \looseness=-1PLaNT.

The Hybrid algorithm is a natural fit for scale-free networks. These graphs may have a large 
tree-width $w$ but they exhibit a \textit{core-fringe} structure with a small dense core whose 
removal leaves a fringe like structure with very low tree-width \cite{weiCoreFringe, akibaCoreFringe}. Typical degree and centrality based ordering schemes also tend to rank the 
vertices in the dense core highly. In such graphs, the Hybrid algorithm uses PLaNT to construct 
SPTs rooted at core vertices which generate a large number of labels. SPTs rooted on fringe 
vertices generate few labels and are constructed using DGLL which exploits heavy pruning to limit 
computation. 

For graphs with a  \textit{core-fringe} structure, a relaxed tree decomposition $(\chi, T(V_T, E_T))$ parameterized by an integer $c$ can be computed such that $|X_{t_r}|=w_m \wedge |X_t|\leq c\ \forall\ t\in V_T\setminus t_{r}$, where $t_r$ is the root of $T$ and $\chi=\{X_t \subseteq V\ \forall\ t\in V_T\}$ maps vertices in $T$ (bags) to subset of vertices in $G$\cite{akibaCoreFringe}. In other words, except root bag, $|X_t|$ is bounded by a parameter $\{c | c \ll w \leq w_m\}$.

\begin{lemma}
    The hybrid algorithm performs $\mathcal{O}(m\log{n}\cdot(w_m + c\log^2{n}) + nc\log{n}\cdot(w_m + c\log{n}))$ work, broadcasts only $\mathcal{O}(cn\log{n})$ data, generates $\mathcal{O}(n\cdot(w_m + c\log{n}))$ hub labels and answers each query in $\mathcal{O}(w_m + c\log{n})$ time.
\end{lemma}
\begin{proof}
    Consider the relaxed tree decomposition \\$(\chi, T(V_T, E_T))$ with root $t_r$ and perform centroid decomposition on all subtrees rooted at the children of $t_{r}$ to obtain tree $T'$. 
    The height of any tree in the forest generated by removing $t_r$ from $T'$ is $\mathcal{O}(\log{n})$. Hence, the height of $T' = \mathcal{O}(\log{n}+1) = \mathcal{O}(\log{n})$.
    
    Consider a ranking $R$ where $R(v)$ is determined by the minimum depth bag $\{b\in V_{T'} | v\in X_b\}$. For GLL, the number of labels generated by SPTs from vertices in root bag is $O(w_mn)$. Combining this with lemma \ref{lemma:complexityLCC}, we can say that total labels generated by GLL is $\mathcal{O}(n\cdot(w_m + c\log{n}))$ and query complexity is $\mathcal{O}(w_m + c\log{n})$. The same also holds for the Hybrid algorithm since it outputs the same CHL as \looseness=-1GLL.
    
    If Hybrid algorithm constructs $w_m$ SPTs using PLaNT and rest using DGLL, the overall work-complexity is $\mathcal{O}(w_m\cdot (m + n\log{n})) + \mathcal{O}(mc\log^2{n} + nc\log{n}\cdot(w_m + c\log{n})) = \mathcal{O}((m\log{n}\cdot(w_m + c\log^2{n}) + nc\log{n}\cdot(w_m + c\log{n})))$. 
    
    The Hybrid algorithm only communicates $\mathcal{O}(cn\log{n})$ labels generated after switching to DGLL, resulting in \\$\mathcal{O}(cn\log{n})$ data broadcast. In comparison, doing only DGLL for the same ordering will broadcast $\mathcal{O}(w_mn+cn\log{n})$ data.
\end{proof}

\noindent{In reality, we use the ratio $\Psi$ as a heuristic, dynamically switching from PLaNT to DGLL when $\Psi$ becomes greater than a 
\looseness=-1threshold $\Psi_{th}$.}

\begin{lemma}
    The Hybrid algorithm consumes\\ $\mathcal{O}\left(\frac{n\cdot(w_m + c\log{n})}{q} + n + m\right)$ main memory per node, where $q$ is the number of nodes used.
\end{lemma}
\begin{proof} 
    Distributed labels use $\mathcal{O}\left(\frac{n\cdot(w_m + c\log{n})}{q}\right)$ space per node and storing the graph requires $O(n+m)$ \looseness=-1space. 
\end{proof}

\subsection{Enabling efficient Multi-node pruning}\label{sec:commTable}
We propose an optimization that \textit{simultaneously} solves the following two problems:\\
\textit{\underline{1. Pruning traversal in PLaNT}} $\rightarrow$ 
The reason why PLaNT cannot prune using rank or distance queries is that
with pruning using partial label info, an SPT can still visit those vertices which would've been pruned if \textit{all} prior labels were available and possibly, through non shortest paths with the wrong ancestor information. This can lead to redundant label generation and defeat the purpose of PLaNT. 

In general, if a node prunes using $H_u$, it must have $\{H_v \forall v\in V | R(v)\geq R(u)\}$ to guarantee non-redundant labels. When this condition is met, a vertex is either visited through the shortest path with correct ancestor or is pruned. We skip the proof details for \looseness=-1brevity.\\
\textit{\underline{2. Redundant labels in DGLL}} $\rightarrow$ Fig.\ref{fig:pruneVsLabel} shows the label count generated by PLL if pruning queries are restricted to use hub labels from few top-ranked hubs only. We observe that label count decreases dramatically even if pruning utilizes only few highest-ranked \looseness=-1hubs. 

\begin{figure}[]
    \centering
 \subfloat{\includegraphics[width=0.5\linewidth]{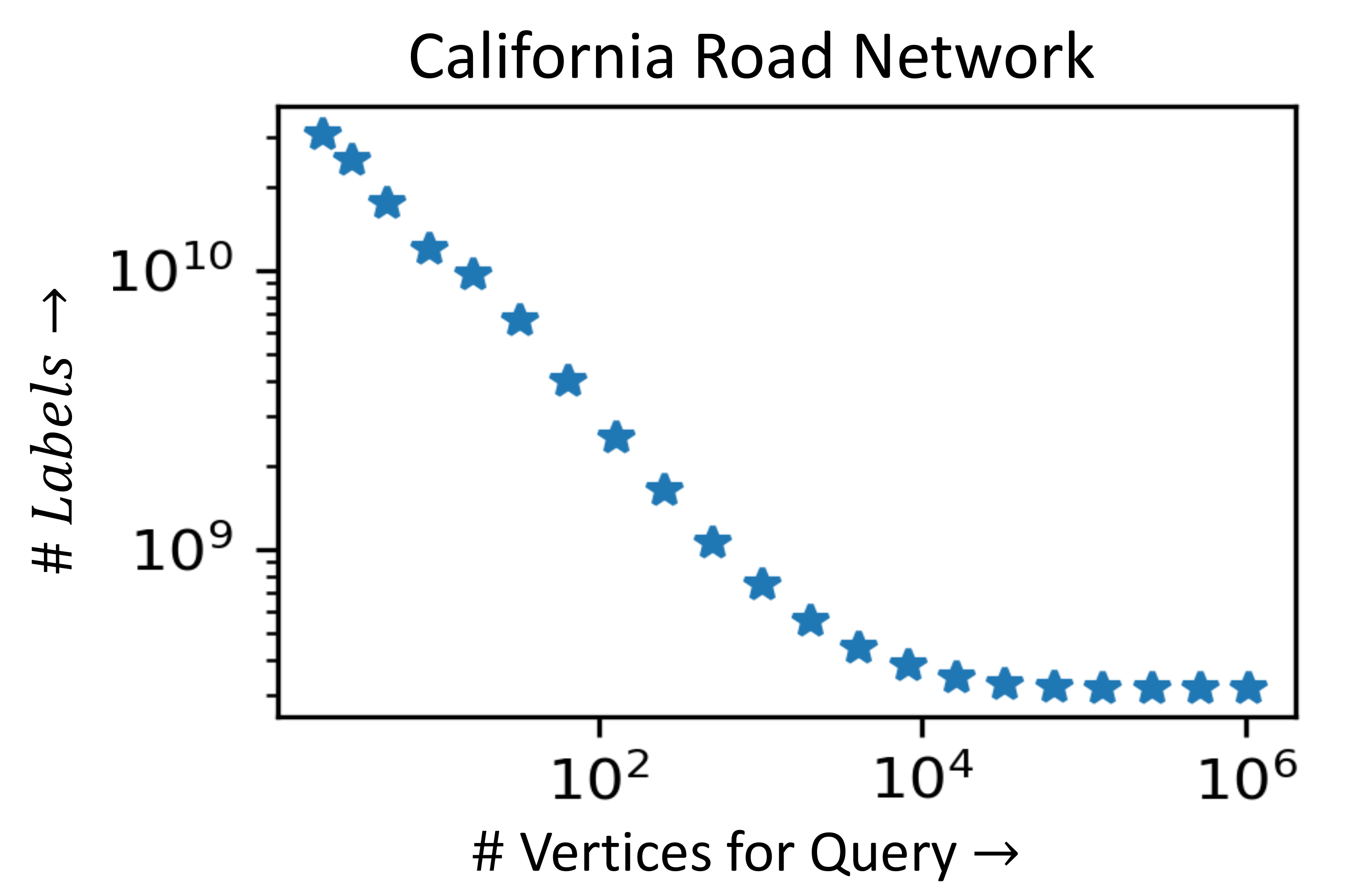}}~~ \subfloat{\includegraphics[width=0.5\linewidth]{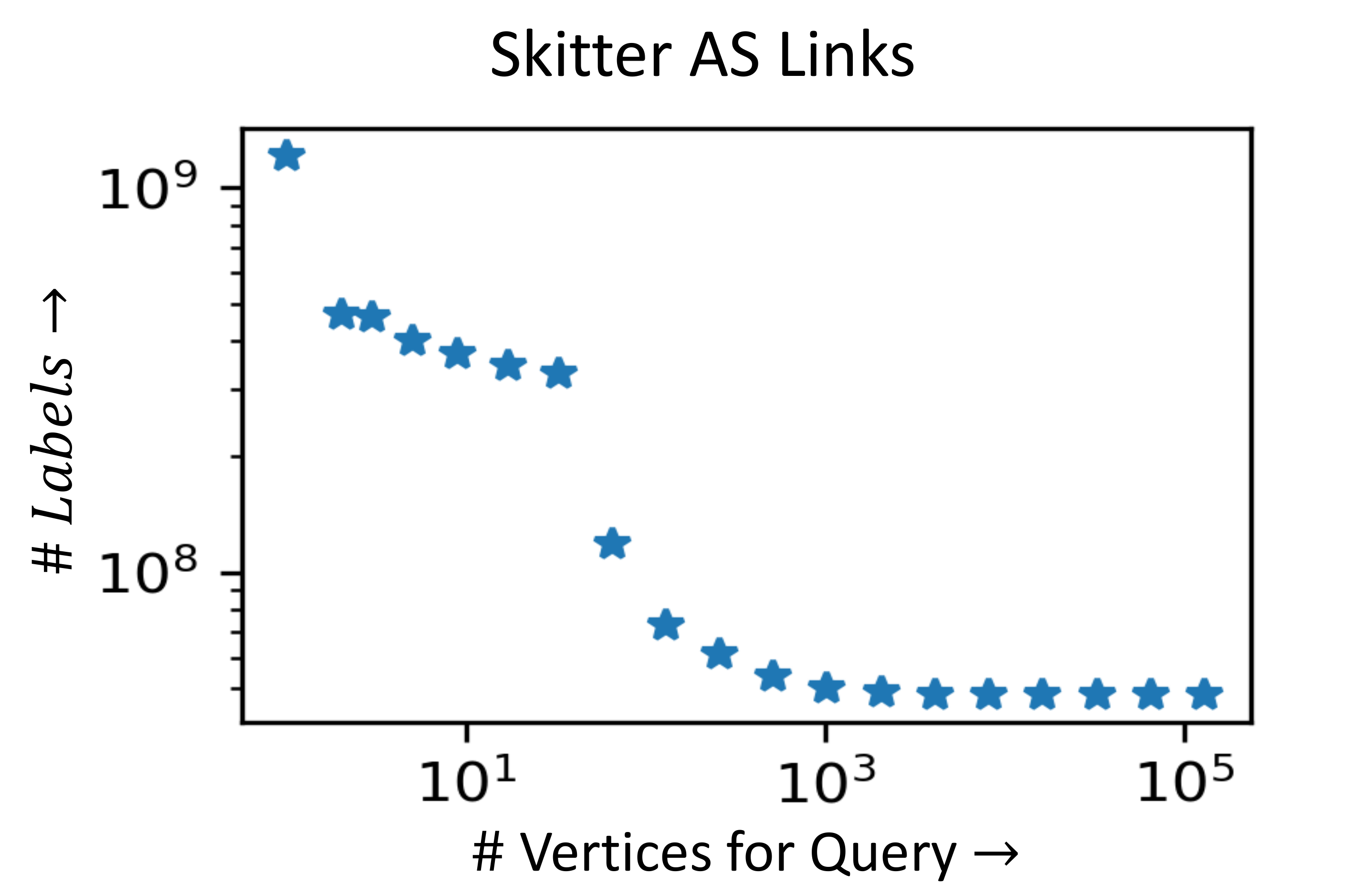}}
    \caption{\# Labels generated if pruning queries in PLL use few (x-axis) highest rank hubs. X-axis$=0$ means rank queries only. When pruning is completely absent, \# labels = $|V|^2$\vspace{-2.5mm}}
    \label{fig:pruneVsLabel}
\end{figure}

Thus, for a given integer $\eta$, if we store \textit{all} labels from $\eta$ most important hubs on every compute node i.e. $HC=\cup_{v\in V | R(v)\geq n-\eta}\{H_v\}$, we \looseness=-1can
\begin{itemize}[leftmargin=*]
    \item use distance queries on $HC$ to prune PLaNTed trees, and
    \item drastically increase pruning efficiency of DGLL.
\end{itemize}

To this purpose, we allocate a \textit{Common Label table} on every node that stores common labels $HC$. These labels are broadcasted even if they are generated by PLaNT. 
For $\eta=\mathcal{O}(1)$, using common labels incurs additional $\mathcal{O}(n)$ broadcast traffic, $\mathcal{O}(w_mn)$ queries of $\mathcal{O}(1)$ complexity each
, and consumes $O(n)$ more memory per node. Thus, it does not alter the theoretical bounds on work-complexity, communication volume and space requirements of the Hybrid \looseness=-1algorithm. In our implementation, we store labels from $\eta =16$ highest ranked hubs in the Common Label Table. 





\subsection{Extensions}
The ideas discussed in section \ref{sec:smp} and \ref{sec:dmp} are not restricted to multi-node clusters and can be used for any massively parallel system, such as GPU. On GPUs, simply parallelizing PLL is not enough, because the first tree constructed by every concurrent thread will not have any label information from higher ranked SPTs and will not prune at all on distance queries. For a GPU which can run thousands of concurrent threads, this can lead to an unacceptable blowup in label size making Label Cleaning extremely time consuming. Even worse, the system may simply run out of memory. Instead, we can use PLaNT to construct first few SPTs for every thread and switch to GLL afterwards. Our approach can also be extended to disk-based processing where access cost to labels is very high. The Common Label Table can be mapped to faster memory in the hierarchy (DRAM) to accelerate pre-processing. 
Finally, we note that by storing the parent of each vertex in an SPT along with the corresponding hub label, CHL can also be used to compute \textit{shortest paths} in time linear to the number of edges in the paths. 

%% file: query.tex
\section{Querying}\label{sec:query}
We provide three modes to the user for distance queries:
\begin{itemize}[leftmargin=*]
    \item \textit{Querying with Labels on Single Node (QLSN)} $\rightarrow$ In this mode, all labels are stored on every node and a query response is computed only by the node where the query emerges. Existing hub labeling frameworks \cite{parapll, akibaPLL, dongPLL, parallelPLLThesis} only support this mode.
    \item \textit{Querying with Fully Distributed Labels (QFDL)} $\rightarrow$ In this mode, the label set of every vertex is partitioned between all nodes
    . A query is broadcasted to all nodes and individual responses of the nodes are reduced using MPI\_MIN to obtain the shortest path distance. It utilizes parallel processing power of multiple nodes and consumes only $\mathcal{O}\left(\frac{|V|\cdot(w_m + c\log{|V|})}{q}\right)$ memory per node, but incurs high communication costs.
    \item \textit{Querying with Distributed Overlapping Labels (QDOL)} $\rightarrow$ In this mode, 
    we divide the vertex set $V$ into $\zeta$ partitions. For every possible partition pair, a node is allocated that stores complete label set of all vertices in that pair. Thus, a given query is answered completely by a single node but not by every node. 
    Unlike QFDL, this mode utilizes the more efficient P2P communication instead of broadcasting. Each query $(u,v)$ is mapped to the node that has labels for vertex partitions containing $u$ and $v$. The query is then communicated to this node which single-handedly computes and sends back the response. In this case, multi-node parallelism is exploited in a batch of queries where different nodes simultaneously compute responses to the respective queries mapped to them.
    
    For a cluster of $q$ nodes, $\zeta$ can be computed as follows: 
    \begin{equation*}
        {\zeta\choose 2} = q \implies \zeta = \frac{1 + \sqrt{1+8q}}{2}
    \end{equation*}
    Storing labels of two vertex partitions consumes\\ $\frac{2|V|\cdot(w_m + c\log{|V|})}{\zeta} = \mathcal{O}\left(\frac{|V|\cdot(w_m + c\log{|V|})}{\sqrt{q}}\right)$ memory per node (much larger than \looseness=-1QFDL). 
\end{itemize}


%% file: experiment.tex
\section{Experiments}
\subsection{Setup}
We conduct shared-memory experiments on a 36 core, 2-way hyperthreaded, dual-socket linux server with two Intel Xeon E5-2695 v4 processors@ 2.1GHz and 1TB DRAM. For the distributed memory experiments, we use a 64-node cluster with each node having an 8 core, 2-way hyperthreaded, Intel Xeon E5-2665@ 2.4GHz processor and 64GB \looseness=-1DRAM. We use OpenMP v4.5 and OpenMPI v3.1.2 for \looseness=-1parallelization. Our shared-memory implementations use all $36$ 
 cores with hyperthreading and distributed implementations use all $8$ cores with hyperthreading on each node. 

\textbf{\textit{Baselines:}} We use sequential PLL (seqPLL), state-of-the-art paraPLL shared-memory (SparaPLL) and distributed (DparaPLL) versions for comparison of pre-processing efficiency. We enable SparaPLL with dynamic task assignment policy for good load balancing. Our implementation of DparaPLL\footnote{paraPLL code is not publicly available.} executes SparaPLL on every compute node
using a task queue with circular allocation (section \ref{sec:dgll}).
We observed that DparaPLL scales poorly as the number of compute nodes increase. This is 
because of high communication overhead
and label size explosion that can be attributed to the absence of rank queries. 
Therefore, we also plot the performance of DGLL for better baselining.
Both DGLL and DparaPLL implementation synchronize $\log_8{n}$ times to exchange labels among the nodes.

\begin{figure}[htb]
    \centering
\includegraphics[width=0.9\linewidth]{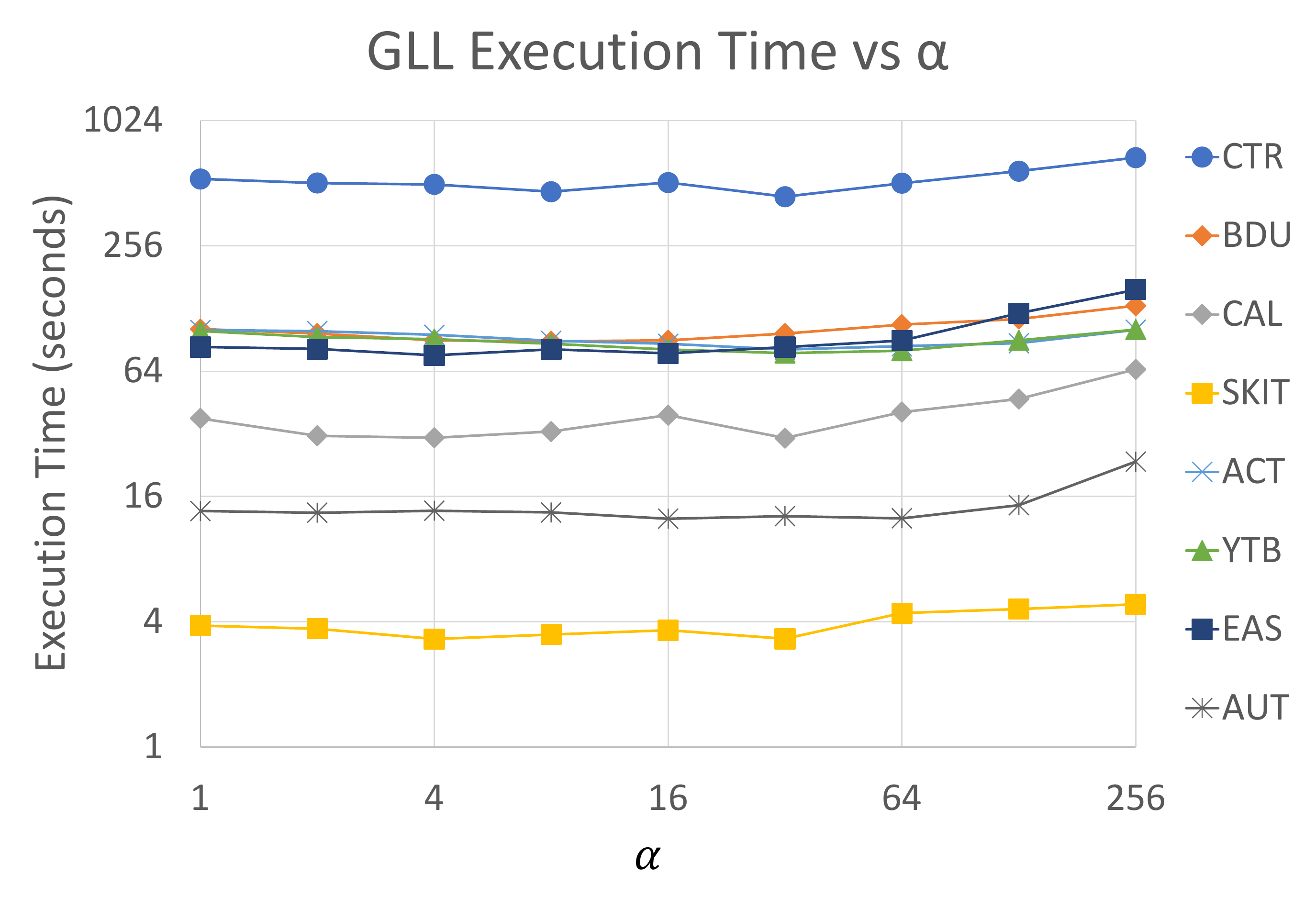}
\caption{\vspace{-1mm}GLL execution time vs synchronization threhsold $\alpha$\vspace{-1mm}}
    \label{fig:alpha}
\end{figure}

\textbf{\textit{Implementation Details:}}
We vary the synchronization threshold $\alpha$ in GLL and switching threshold $\Psi_{th}$ in the 
Hybrid algorithm to empirically assess the impact of these parameters on algorithm performance.
Figure \ref{fig:alpha} shows the impact of $\alpha$ on GLL. We note that the execution time is robust to significant variations in $\alpha$ within a range of $2$ to $32$. Intuitively, a small value of 
$\alpha$ reduces cleaning time (section \ref{sec:gll}) but making it too small can lead to frequent
synchronizations that hurt parallel performance.
Based on
our observations, we set $\alpha=4$ for further experiments. 
Figure \ref{fig:psi} shows the effect of $\Psi_{th}$ on the performance of hybrid algorithm. Intuitively, keeping $\Psi_{th}$ too large increases the computation overhead (seen in scale-free networks) because even low-ranked SPTs that generate few labels, are PLaNTed. On the other hand, keeping 
$\Psi_{th}$ too small results in poor scalability (seen in road networks) as the algorithm switches to DGLL quickly and parallelism and communication avoidance of PLaNT remain underutilized. Based on these findings, we set $\Psi_{th}=100$ for scale-free networks and $\Psi_{th}=500$ for road networks. 

\begin{figure}[htb]
    \centering
 \subfloat{\includegraphics[width=0.5\linewidth]{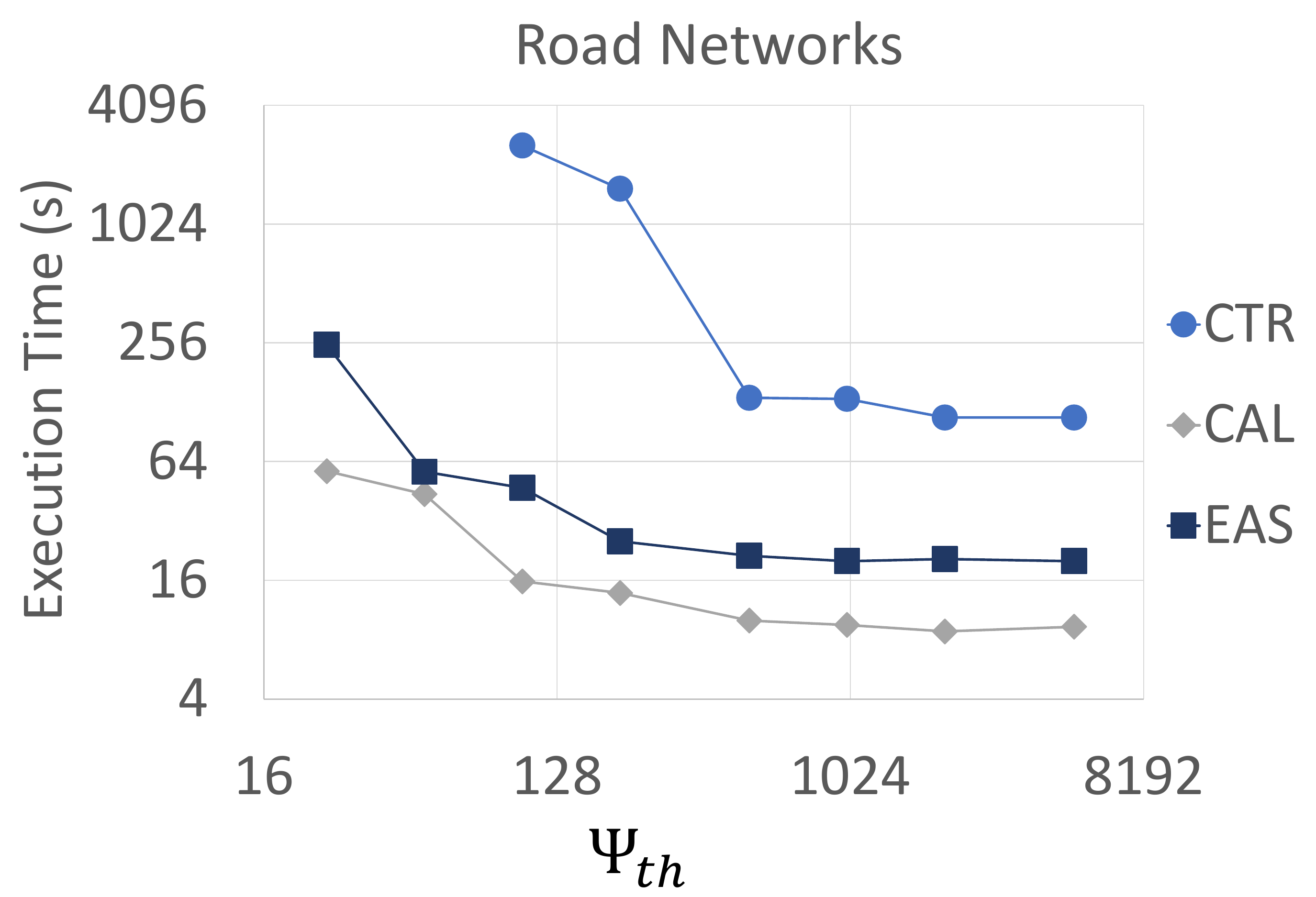}}
 ~
 \subfloat{\includegraphics[width=0.5\linewidth]{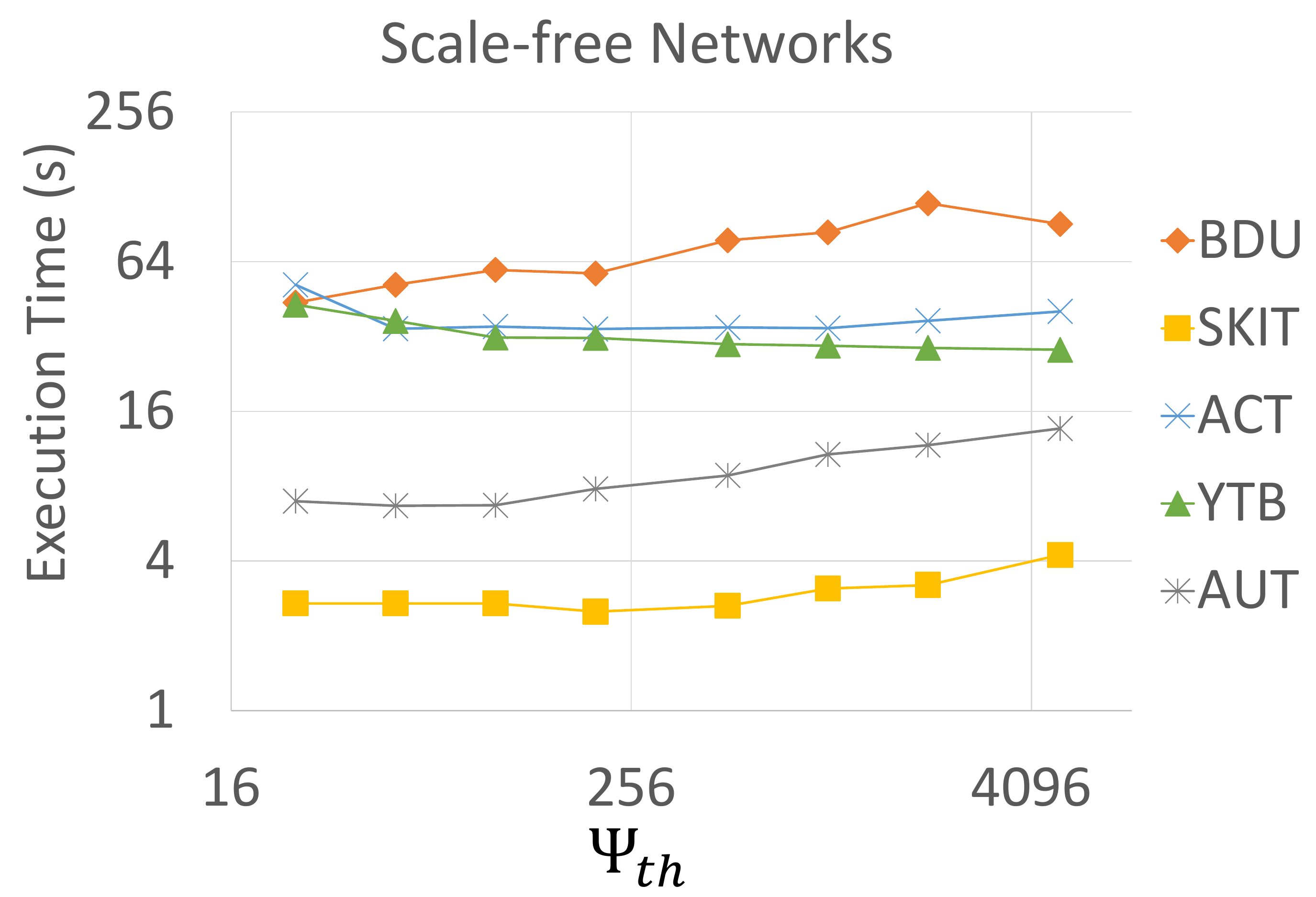}}
 
    \caption{Execution time of Hybrid algorithm on 16 compute nodes vs switching threshold $\Psi_{th}$\vspace{-2mm}}
    \label{fig:psi}
\end{figure}



\subsubsection{Datasets}
We evaluate our algorithms on $12$ real-world graphs with varied topologies, as listed in table \ref{table:datasets}. The scale-free networks do not have edge weights from the download sources. For such graphs, we assign edge weights between $[1, \sqrt{n})$ uniformly at random.

The ranking $R$ is determined by betweenness for road networks\cite{vldbExperimental} and degree 
for scale-free networks\cite{akibaPLL}. Betweenness is approximated by sampling a few shortest path trees and both methods are inexpensive to compute.

\begin{table}[htbp]
\caption{Datasets for Evaluation}
\label{table:datasets}
\resizebox{\linewidth}{!}{%
\begin{tabular}{|c|c|c|c|c|}
\hline
\textbf{Dataset} & \textbf{$n$} & \textbf{$m$} & \textbf{Description} & \textbf{Type} \\ \hline
CAL\cite{roadDimacs}                                        & 1,890,815                                   & 4,657,742                                  & California Road network                        & Undirected                     \\ \hline
EAS\cite{roadDimacs}                                          & 3,598,623                                   & 8,778,114                                  & East USA Road network                          & Undirected                     \\ \hline
CTR\cite{roadDimacs}                                        & 14,081,816                                  & 34,292,496                                 & Center USA Road network                        & Undirected                     \\ \hline
USA\cite{roadDimacs}                                        & 23,947,347                                  & 58,333,344                                 & Full USA Road network                          & Undirected                     \\ \hline
SKIT\cite{skitter}                                       & 192,244                                     & 636,643                                    & Skitter Autonomous Systems               & Undirected                     \\ \hline
WND\cite{und}                                        & 325,729                                     & 1,497,134                                  & Univ.  Notre Dame webpages                 & Directed                       \\ \hline
AUT\cite{coauth}                                       & 227,320                                     & 814,134                                    & Citeseer Collaboration                & Undirected                       \\ \hline
YTB\cite{konect}                                        & 1,134,890                                   & 2,987,624                                  & Youtube Social network              & Undirected                     \\ \hline
ACT\cite{konect}                                        & 382,219                                     & 33,115,812                                 & Actor Collaboration Network                    & Undirected                     \\ \hline
BDU\cite{konect}                                        & 2,141,300                                   & 17,794,839                                 & Baidu HyperLink Network                        & Directed                       \\ \hline
POK\cite{konect}                                        & 1,632,803                                   & 30,622,564                                 &  Social network Pokec                 & Directed                       \\ \hline
LIJ\cite{konect}                                        & 4,847,571                                   & 68,993,773                                 & LiveJournal Social network                     & Directed                       \\ \hline
\end{tabular}}
\end{table}

\subsection{Evaluation of Shared-memory Algorithms}
Table \ref{table:smp} compares the performance of GLL with LCC, SparaPLL and seqPLL. It also shows the Average Label Size (ALS) per vertex in CHL (GLL, LCC and seqPLL) and labeling generated by SparaPLL. The query response time is directly proportional to Average Label Size (ALS) per vertex and hence, ALS is a crucial parameter for any hub labeling algorithm.
In case of LIJ graph, none of the shared-memory algorithms finished execution and its CHL ALS was obtained from the distributed algorithms.


We observe that on average, GLL generates $17\%$ less labels than paraPLL which can be quite significant for an application that generates many PPSD queries. GLL is only $1.15\times$ slower than paraPLL on average even though it re-checks every label generated and the cleaning queries use linear-merge based querying\footnote{For space efficiency, the labels are only stored as $L_v$ (ordered by vertex) and there is no copy of labels stored as $H_v$ (ordered by hubs).} as opposed to the more efficient hash-join label construction queries. 

For some graphs such as CAL, GLL is even $1.3\times$ faster than SparaPLL. This is because of rank queries, faster label construction queries due to smaller sized label sets and lock avoidance in GLL. Although not shown in table \ref{table:smp} for brevity, we observed that ALS and scalability of paraPLL worsen as number of threads increase. Hence, we expect the relative performance of GLL to be even better on systems with more \looseness=-1parallelism.

Fig. \ref{fig:breakup} shows execution time breakup for LCC and GLL. GLL cleaning is significantly faster than LCC because of the reduced cleaning complexity (section \ref{sec:gll}). Overall, GLL is $1.25\times$ faster than LCC on average. However, for some graphs such as CAL, fraction of cleaning time is $>30\%$ even for GLL. This is because in the first superstep of GLL, number of labels generated is more than $\alpha n$ as there are no labels available for distance query pruning and number of simultaneous SPTs under construction is $p>\alpha$ ($p$ is \# threads). This problem can be circumvented by using PLaNT for the first superstep in shared-memory implementation as well. 

\begin{table}[htbp]
\caption{Performance comparison of GLL and LCC with baselines. ALS is the average label size per vertex and time=$\infty$ implies execution did not finish in $2$ hours}
\label{table:smp}
\resizebox{\linewidth}{!}{%
\begin{tabular}{c|c|c|c|c|c|c|}
\cline{2-7}
                                       & \multicolumn{2}{c|}{\textbf{SparaPLL}} & \multirow{2}{*}{\textbf{\begin{tabular}[c]{@{}c@{}}CHL \\ ALS\end{tabular}}} & \multirow{2}{*}{\textbf{\begin{tabular}[c]{@{}c@{}}seqPLL \\ Time(s)\end{tabular}}} & \multirow{2}{*}{\textbf{\begin{tabular}[c]{@{}c@{}}LCC \\ Time(s)\end{tabular}}} & \multirow{2}{*}{\textbf{\begin{tabular}[c]{@{}c@{}}GLL \\ Time(s)\end{tabular}}} \\ \cline{1-3}
\multicolumn{1}{|c|}{\textbf{Dataset}} & \textbf{ALS}     & \textbf{Time(s)}    &                                                                              &                                                                                     &                                                                                   &                                                                                   \\ \hline
\multicolumn{1}{|c|}{CAL}              & $108.3$          & $51.2$              & $83.4$                                                                       & $215$                                                                               & $41.4$                                                                            & $35.4$                                                                            \\ \hline
\multicolumn{1}{|c|}{EAS}              & $138.1$          & $116.3$             & $116.8$                                                                      & $680.6$                                                                             & $108.7$                                                                           & $88$                                                                              \\ \hline
\multicolumn{1}{|c|}{CTR}              & $178.7$          & $424.2$             & $160.9$                                                                      & $5045$                                                                              & $664.1$                                                                           & $567.7$                                                                           \\ \hline
\multicolumn{1}{|c|}{USA}              & $185.6$          & $816.9$             & $166.1$                                                                      & $\infty$                                                                            & $1148.6$                                                                          & $834$                                                                             \\ \hline
\multicolumn{1}{|c|}{SKIT}             & $88.3$           & $2.5$               & $85.1$                                                                       & $95.8$                                                                              & $4.85$                                                                             & $3.9$                                                                             \\ \hline
\multicolumn{1}{|c|}{WND}              & $39.6$           & $2.4$               & $23.5$                                                                       & $21.98$                                                                             & $2.94$                                                                             & $2.1$                                                                             \\ \hline
\multicolumn{1}{|c|}{AUT}              & $240.2$          & $10.4$              & $229.6$                                                                      & $670$                                                                               & $18.4$                                                                            & $14.6$                                                                            \\ \hline
\multicolumn{1}{|c|}{YTB}              & $208.9$          & $69.6$              & $207.5$                                                                      & $2692.6$                                                                            & $126.7$                                                                             & $104.6$                                                                           \\ \hline
\multicolumn{1}{|c|}{ACT}              & $376.1$          & $112.4$             & $366.3$                                                                      & $\infty$                                                                            & $151.3$                                                                           & $141.9$                                                                           \\ \hline
\multicolumn{1}{|c|}{BDU}              & $100.1$          & $103.1$             & $90.7$                                                                       & $4736$                                                                              & $133.9$                                                                           & $99.9$                                                                            \\ \hline
\multicolumn{1}{|c|}{POK}              & $2243.4$          & $4159.3$             & $2230.7$                                                                       & $\infty$                                                                              & $\infty$                                                                           & $3916.5$                                                                            \\ \hline
\multicolumn{1}{|c|}{LIJ}              & $-$          & $\infty$             & $1222.5$                                                                       & $\infty$                                                                              & $\infty$                                                                           & $\infty$                                                                            \\ \hline
\end{tabular}
}
\end{table}

\begin{figure}[]
    \centering
\includegraphics[width=\linewidth]{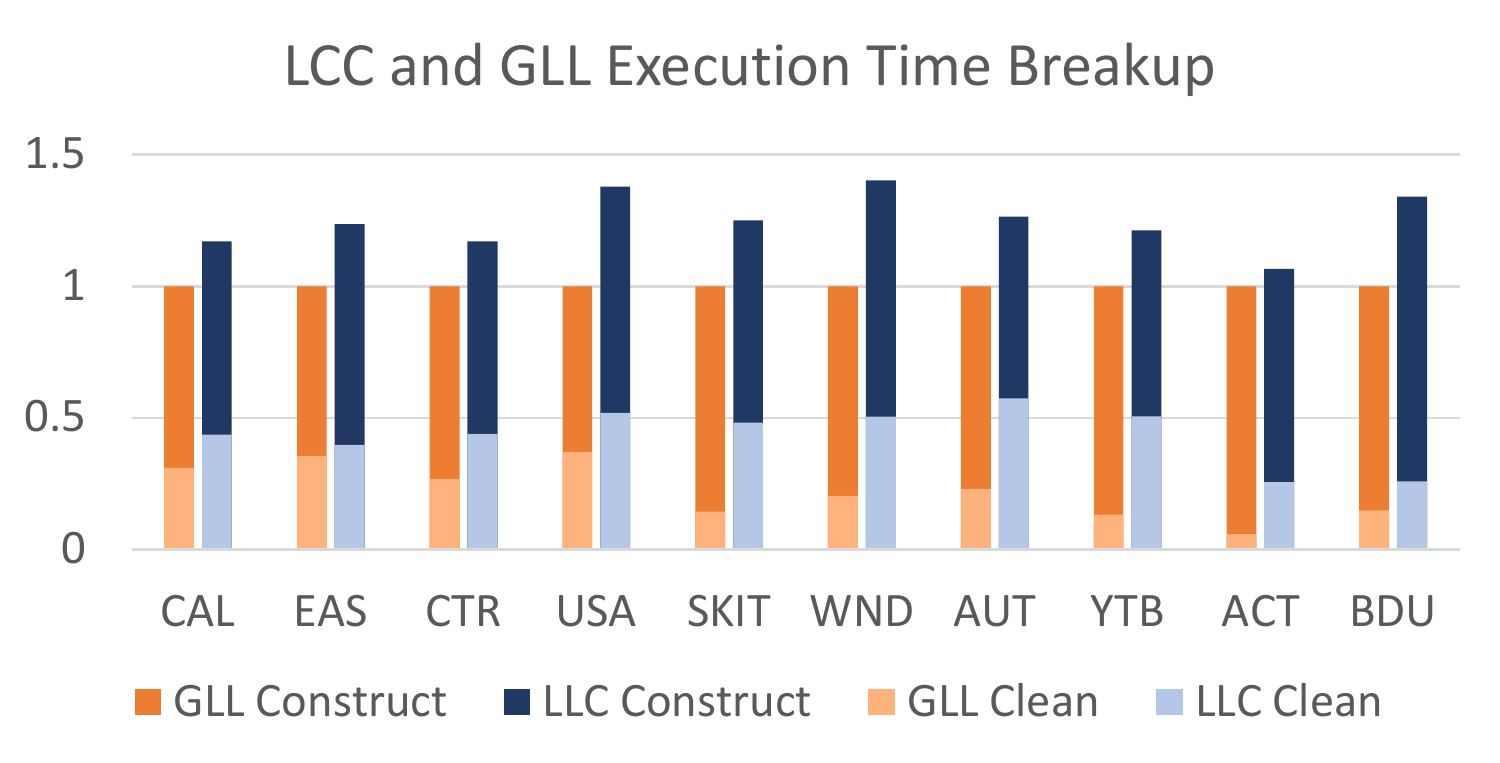}     \caption{Time taken for label construction and cleaning in LCC and GLL normalized by the total execution time of GLL\vspace{-3mm}}
    \label{fig:breakup}
\end{figure}

 \subsection{Evaluation of Distributed-memory Algorithms}
 To assess the scalability of distributed hub labeling algorithms, we vary $q$ from 1 to 64 (\# compute cores from 8 to 512). Fig. \ref{fig:scale} shows the strong scaling of different algorithms in terms of labeling construction time.
 
 \begin{figure*}[htb]
    \centering
\includegraphics[width=\linewidth]{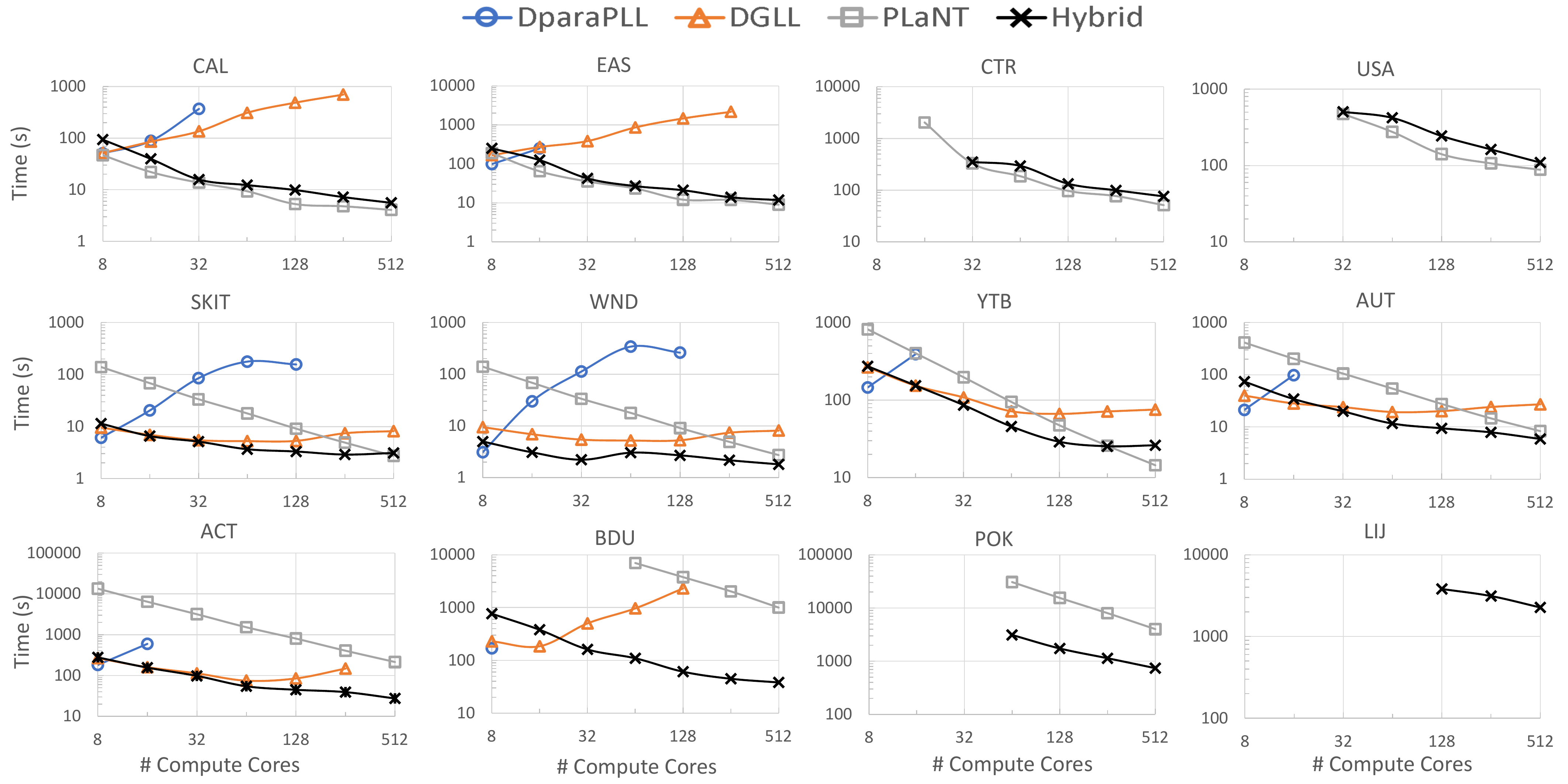}     \caption{Strong scaling results of DparaPLL, DGLL, PLaNT and Hybrid algorithms. Missing curves or points mean that the algorithm incurred OOM error or did not finish within $4$ hours. Also, \# compute cores = $8*$(\# compute nodes).}
    \label{fig:scale}
\end{figure*}
 
 We note that both DparaPLL and DGLL do not scale well as $q$ increases. DparaPLL often runs out-of-memory when $q$ is large. This is because in the first superstep itself, a large number of hub labels are generated that when exchanged, overwhelm the memory of the nodes. DGLL, on the other hand, limits the amount of labels per superstep by synchronizing frequently in the early stage of execution and increasing the synchronization point later. 
 
 Moreover, due to the absence of rank queries, the label size of DparaPLL explodes as $q$ increases (fig.\ref{fig:labelSize}). The efficiency of distance query based pruning in DparaPLL suffers 
 because on every compute node, labels from several high-ranked hubs (that cover a large number 
 of shortest paths) are missing. As the label size explodes, distance queries become expensive and the pre-processing becomes dramatically slower. On the other hand, \textit{rank queries} in DGLL allow pruning even at those hubs whose SPTs were not created on the querying node. Further, it periodically cleans redundant labels, thus, retaining the performance of distance queries. Yet, DGLL incurs significant communication and slows down as more compute nodes are involved in pre-processing. Neither DparaPLL, nor DGLL are able to process the large CTR, USA, POK and LIJ datasets, the former running out-of-memory and the latter failing to finish execution within time
 limit.

\begin{figure}[htb]
    \centering
\includegraphics[width=\linewidth]{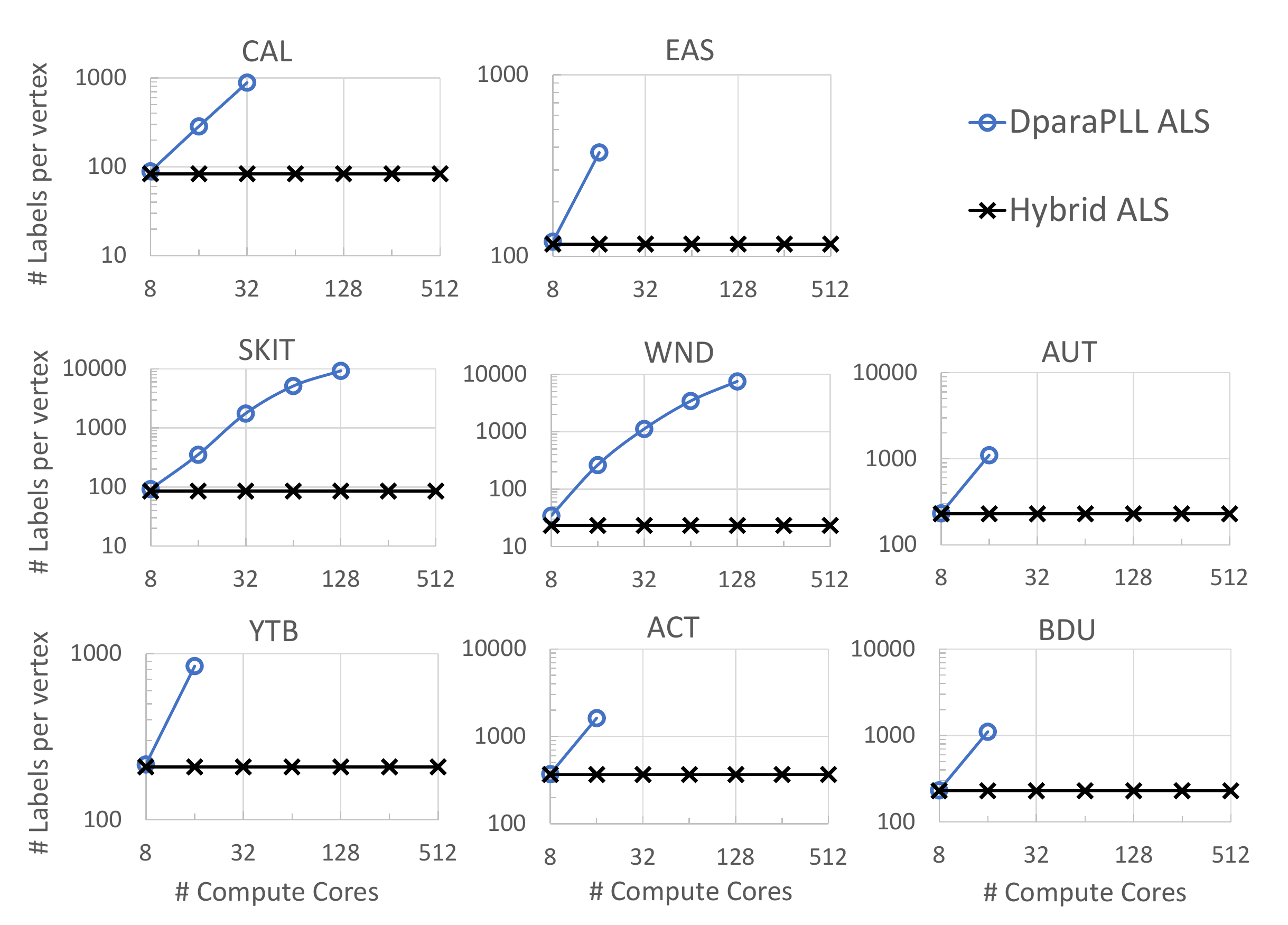}     
\caption{Average Label Size (ALS) generated by DparaPLL and Hybrid algorithms.\vspace{-2mm}}
    \label{fig:labelSize}
\end{figure}
 
 \begin{table*}[htb]
\caption {Query Processing Throughput, Latnecy and Total Memory Consumption for different modes on $16$ compute nodes. "-" means that corresponding mode cannot be supported due to main memory constraints.}
\label{table:queryPerf}
\centering
\resizebox{0.8\linewidth}{!}{
\centering
\begin{tabular}{|c|c|c|c|c|c|c|c|c|c|}
\hline
\multirow{2}{*}{\textbf{Dataset}} & \multicolumn{3}{c|}{\textbf{Throughput (million queries/$s$)}} & \multicolumn{3}{c|}{\textbf{Latency ($\mu s$ per query)}} & \multicolumn{3}{l|}{\textbf{Memory Usage (GB)}} \\ \cline{2-10} 
                                  & \textbf{QLSN}       & \textbf{QFDL}       & \textbf{QDOL}      & \textbf{QLSN}     & \textbf{QFDL}     & \textbf{QDOL}     & \textbf{QLSN}  & \textbf{QFDL}  & \textbf{QDOL} \\ \hline
CAL                               & 10.1                & 12.1                & \textbf{29.6}      & \textbf{2.8}      & 22.3              & 8.4               & 43.8           & \textbf{2.4}   & 13.7          \\ \hline
EAS                               & 7.1                 & 8.9                 & \textbf{14.6}      & \textbf{3.6}      & 24                & 11.4              & 125.4          & \textbf{7.4}   & 39.2          \\ \hline
CTR                               & -                   & 6.5                 & \textbf{9}      & -                 & 26.6              & \textbf{14.7}     & -              & \textbf{45}    & 242.1         \\ \hline
USA                               & -                   & 5.4                 & \textbf{10}        & -                 & 29.5              & \textbf{20}       & -              & \textbf{80}    & 413.3         \\ \hline
SKIT                              & 15.8                & 18.5                & \textbf{29.8}      & \textbf{1}        & 20.7              & 7.9               & 4.5            & \textbf{0.3}   & 1.4           \\ \hline
WND                               & 37.5                & 19.6                & \textbf{42.7}      & \textbf{0.3}      & 22.7              & 7.1               & 0.6            & \textbf{0.1}  & 0.6           \\ \hline
AUT                               & 4.9                 & 9.9                 & \textbf{27.5}      & \textbf{3.7}      & 21.7              & 12.9              & 16.6           & \textbf{1}     & 5.2           \\ \hline
YTB                               & 10.4                & 23.3                & \textbf{30.3}      & \textbf{2.2}      & 23.9              & 13.6              & 74.9           & \textbf{4.6}   & 23.4          \\ \hline
ACT                               & 3.2                 & 10.4                & \textbf{21.3}      & \textbf{4.8}      & 22.8              & 18.1              & 46.1           & \textbf{2.8}   & 14.4          \\ \hline
BDU                               & 13.2                & 16.4                & \textbf{21.5}      & \textbf{1.5}      & 22.1              & 11.1              & 54.7           & \textbf{3.2}   & 17.1          \\ \hline
POK                               & -                   & 5.1                 & \textbf{7.5}       & -                 & \textbf{32}                & 34.5     & -              & \textbf{77.6}  & 388.9         \\ \hline
LIJ                               & -                   & \textbf{6}          & -                  & -        & \textbf{31.6}     & -                 & -              & \textbf{125.8} & -             \\ \hline
\end{tabular}

}
\end{table*}

 PLaNT on the other hand, paints a completely different picture. Owing to its embarrassingly parallel nature, PLaNT exhibits excellent near-linear speedup upto $q=64$ for almost all datasets. On average, PLaNT is able to achieve $42\times$ speedup on $64$ nodes compared to single node execution. However, for scale-free graphs, PLaNT is not efficient. It is unable to process LIJ and takes more than an hour to process POK dataset even on $64$ nodes.
 
 The Hybrid algorithm combines the scalability of PLaNT with the pruning efficiency of DGLL (powered by Common Labels). It scales well upto $q=64$ and for most datasets, achieves $>10\times$ speedup over single node execution. At the same time, for large scale-free datasets ACT, BDU and POK, it is able to construct CHL $7.8\times$, $26.2\times$ and $5.4\times$ faster than PLaNT, respectively,  on $64$ nodes. When processing scale-free datasets on small number of compute nodes ($1$, $2$ or $4$ nodes), Hybrid beats PLaNT by more than an order of magnitude difference in execution time. Compared to DparaPLL, the Hybrid algorithm is $3.8\times$ faster on average when run on $2$ compute nodes. For SKIT and WND, the Hybrid algorithm is $47\times$ and $96.8\times$ faster, respectively, than DparaPLL on $16$ nodes. 
 
 Although fig.\ref{fig:labelSize} only plots ALS for Hybrid algorithm, even PLaNT and DGLL generate the same CHL and hence, have the same label size for any dataset irrespective of $q$ (section \ref{sec:dmp}). We also observe superlinear speedup in some cases (Hybrid $\rightarrow$ CAL and AUT $1$ node vs $2$ nodes; PLaNT $\rightarrow$ CTR $2$ nodes vs $4$ nodes). This is because these datasets generate huge amount of labels that exert a lot of pressure on the memory. When running on few nodes, main memory utilization on every node is very high, slowing down the memory allocation operations that happen when new labels are appended. In such cases, increasing number of nodes not only increases compute cores but also releases pressure on memory management system, resulting in a super linear \looseness=-1speedup.

 \textbf{Graph Topologies:} We observe that PLaNT alone not only scales well but is also extremely efficient for \textit{road networks}. On the other hand, in scale-free networks, PLaNT although scalable is not efficient as it incurs large overhead of additional exploration in low-ranked SPTs. This is consistent with our observations in figure \ref{fig:twVsLabel} where the maximum value of $\Psi$ for SKIT was $>10\times$ that of maximum $\Psi$ in CAL road network. Consequently, the Hybrid algorithm that cleverly manages the trade-off between additional exploration and communication avoidance, is significantly faster than PLaNT for most scale-free networks.

 
 We also observe that the Hybrid algorithm does not scale equally well for small datasets when the number of compute nodes is high. PLaNT eventually catches up with the Hybrid, even beating it in several cases. This is because even few synchronizations of large number of nodes completely dominate the small pre-processing time. Perhaps, scalability for small datasets can be improved by making the number of synchronizations and switching point from PLaNT to DGLL, a function of both $q$ and $\Psi$. 
 
 

\subsection{Evaluation of Query Modes}
In this section, we assess the different query modes on the basis of their memory consumption, query response latency and query processing throughput. Table \ref{table:queryPerf} shows the memory consumed by label storage under different modes. QLSN requires all 
labels to be stored on every node and is the most memory hungry mode. Both QDOL and QFDL distribute the labels
across multiple nodes enabling queries on large graphs where QLSN fails. 
Our experiments also confirm the theoretical insights into the memory usage of QFDL and QDOL presented in section \ref{sec:query}. On average, QDOL requires $5.3\times$ more main memory for label storage than QFDL. This is because the label size
per partition in QDOL scales with $\mathcal{O}\left(\frac{1}{\sqrt{q}}\right)$ and every compute node has to further store label set 
of $2$ such partitions. 

To evaluate the latency of various query modes, we generate $1$ million PPSD queries and compute their
response one at a time. In QFDL (QDOL) mode, one query is transmitted per MPI\_Broadcast (MPI\_Send,
respectively) and inter-node communication latency becomes a major contributor to the overall query response latency.
This is evident from the experimental results (table \ref{table:queryPerf}) where 
latency of QFDL shows little variation across different datasets. Contrarily, QLSN does not incur inter-node
communication and compared to QDOL and QFDL, has significantly lower latency although it increases proportionally with ALS. For most datasets, QDOL latency is $<2\times$ compared to QFDL, because of the cheaper point-to-point 
communication as opposed to more expensive broadcasts (section \ref{sec:query}). An exception is POK, where average label size is huge (table \ref{table:smp}) and QFDL takes advantage of multi-node parallelism to 
reduce \looseness=-1latency.

To evaluate the query throughput, we create a batch of $100$ million PPSD queries and compute their
responses in parallel. For most datasets, the added multi-node parallelism of QFDL and QDOL\footnote{In QDOL 
mode, prior to communicating the queries, we sort them based on the nodes that they map to. After receiving query responses, we rearrange them in original order. 
The throughput reported in table \ref{table:queryPerf} also takes into account,
the time for sorting and rearranging.} overcomes the query communication overhead and results in higher 
throughput than QLSN. QDOL is further $1.8\times$ faster than QFDL
because of reduced communication overhead\footnote{QDOL also has better memory access locality as every node traverses all hub labels of vertices in queries assigned to it. Contrarily, each node in QFDL only traverses a part of hub 
labels for all queries, frequently jumping from one vertex' labels to another.}.

%% file: conclusion.tex
\section{CONCLUSION AND FUTURE WORK}
In this paper, we address the problem of efficiently constructing Hub Labeling and answering shortest distance qu-\\eries on shared and distributed memory parallel systems. We outline the multifaceted challenges associated with the algorithm in general, and specific to the parallel processing platforms. We propose novel algorithmic innovations and optimizations that systematically resolve these challenges.  Our embarassingly parallel algorithm PLaNT, dramatically increases the scalability 
of hub labeling, making it feasible to utilize the massive parallelism in a cluster of compute \looseness=-1nodes.


We show that our approach exhibits good theoretical and empirical performance. Our algorithms are able to scale significantly better than the existing approaches, with orders of magnitude faster pre-processing and capability to process very large graphs.

There are several interesting directions to pursue in the context of this work. We will explore the use of distributed atomics and RMA calls to dynamically allocate tasks even on multiple nodes. This can improve load balance across nodes and further boost the performance of PLaNT and Hybrid algorithms. Another important area for research is development of heuristics to compute switching point between PLaNT and DGLL and to autotune the parameters $\alpha$ and $\beta$ to adapt to the graph topology and number of compute nodes $q$.

%% file: main.bbl
\begin{thebibliography}{10}

\bibitem{skitter}
The skitter as links dataset, 2019.
\newblock [Online; accessed 8-April-2019].

\bibitem{abrahamCHL}
I.~Abraham, D.~Delling, A.~V. Goldberg, and R.~F. Werneck.
\newblock Hierarchical hub labelings for shortest paths.
\newblock In {\em European Symposium on Algorithms}, pages 24--35. Springer,
  2012.

\bibitem{akibaPLL}
T.~Akiba, Y.~Iwata, and Y.~Yoshida.
\newblock Fast exact shortest-path distance queries on large networks by pruned
  landmark labeling.
\newblock In {\em Proceedings of the 2013 ACM SIGMOD International Conference
  on Management of Data}, pages 349--360. ACM, 2013.

\bibitem{akibaCoreFringe}
T.~Akiba, C.~Sommer, and K.-i. Kawarabayashi.
\newblock Shortest-path queries for complex networks: exploiting low tree-width
  outside the core.
\newblock In {\em Proceedings of the 15th International Conference on Extending
  Database Technology}, pages 144--155. ACM, 2012.

\bibitem{und}
R.~Albert, H.~Jeong, and A.-L. Barab{\'a}si.
\newblock Internet: Diameter of the world-wide web.
\newblock {\em nature}, 401(6749):130, 1999.

\bibitem{cohen2hop}
E.~Cohen, E.~Halperin, H.~Kaplan, and U.~Zwick.
\newblock Reachability and distance queries via 2-hop labels.
\newblock {\em SIAM Journal on Computing}, 32(5):1338--1355, 2003.

\bibitem{roadDimacs}
C.~Demetrescu, A.~Goldberg, and D.~Johnson.
\newblock 9th dimacs implementation challenge--shortest paths.
\newblock {\em American Mathematical Society}, 2006.

\bibitem{julienne}
L.~Dhulipala, G.~Blelloch, and J.~Shun.
\newblock Julienne: A framework for parallel graph algorithms using
  work-efficient bucketing.
\newblock In {\em Proceedings of the 29th ACM Symposium on Parallelism in
  Algorithms and Architectures}, pages 293--304. ACM, 2017.

\bibitem{dongPLL}
Q.~Dong, K.~Lakhotia, H.~Zeng, R.~Karman, V.~Prasanna, and G.~Seetharaman.
\newblock A fast and efficient parallel algorithm for pruned landmark labeling.
\newblock In {\em 2018 IEEE High Performance extreme Computing Conference
  (HPEC)}, pages 1--7. IEEE, 2018.

\bibitem{parallelPLLThesis}
D.~Ferizovic and G.~E. Blelloch.
\newblock Parallel pruned landmark labeling for shortest path queries on
  unit-weight networks.
\newblock 2015.

\bibitem{asynch}
J.~S. Firoz, M.~Zalewski, T.~Kanewala, and A.~Lumsdaine.
\newblock Synchronization-avoiding graph algorithms.
\newblock In {\em 2018 IEEE 25th International Conference on High Performance
  Computing (HiPC)}, pages 52--61. IEEE, 2018.

\bibitem{coauth}
R.~Geisberger, P.~Sanders, and D.~Schultes.
\newblock Better approximation of betweenness centrality.
\newblock In {\em Proceedings of the Meeting on Algorithm Engineering \&
  Expermiments}, pages 90--100. Society for Industrial and Applied Mathematics,
  2008.

\bibitem{jiangDisk}
M.~Jiang, A.~W.-C. Fu, R.~C.-W. Wong, and Y.~Xu.
\newblock Hop doubling label indexing for point-to-point distance querying on
  scale-free networks.
\newblock {\em Proceedings of the VLDB Endowment}, 7(12):1203--1214, 2014.

\bibitem{konect}
J.~Kunegis.
\newblock Konect: the koblenz network collection.
\newblock In {\em Proceedings of the 22nd International Conference on World
  Wide Web}, pages 1343--1350. ACM, 2013.

\bibitem{liveJournal}
J.~Leskovec, K.~J. Lang, A.~Dasgupta, and M.~W. Mahoney.
\newblock Statistical properties of community structure in large social and
  information networks.
\newblock In {\em Proceedings of the 17th international conference on World
  Wide Web}, pages 695--704. ACM, 2008.

\bibitem{liSigmod}
W.~Li, M.~Qiao, L.~Qin, Y.~Zhang, L.~Chang, and X.~Lin.
\newblock Scaling distance labeling on small-world networks.
\newblock In {\em Proceedings of the 2019 International Conference on
  Management of Data}, SIGMOD '19, pages 1060--1077, New York, NY, USA, 2019.
  ACM.

\bibitem{vldbExperimental}
Y.~Li, M.~L. Yiu, N.~M. Kou, et~al.
\newblock An experimental study on hub labeling based shortest path algorithms.
\newblock {\em Proceedings of the VLDB Endowment}, 11(4):445--457, 2017.

\bibitem{galois}
D.~Nguyen, A.~Lenharth, and K.~Pingali.
\newblock A lightweight infrastructure for graph analytics.
\newblock In {\em Proceedings of the Twenty-Fourth ACM Symposium on Operating
  Systems Principles}, pages 456--471. ACM, 2013.

\bibitem{parapll}
K.~Qiu, Y.~Zhu, J.~Yuan, J.~Zhao, X.~Wang, and T.~Wolf.
\newblock Parapll: Fast parallel shortest-path distance query on large-scale
  weighted graphs.
\newblock In {\em Proceedings of the 47th International Conference on Parallel
  Processing}, page~2. ACM, 2018.

\bibitem{ligra}
J.~Shun and G.~E. Blelloch.
\newblock Ligra: a lightweight graph processing framework for shared memory.
\newblock In {\em ACM Sigplan Notices}, volume~48, pages 135--146. ACM, 2013.

\bibitem{weiCoreFringe}
F.~Wei.
\newblock Tedi: efficient shortest path query answering on graphs.
\newblock In {\em Graph Data Management: Techniques and Applications}, pages
  214--238. IGI Global, 2012.

\bibitem{gemini}
X.~Zhu, W.~Chen, W.~Zheng, and X.~Ma.
\newblock Gemini: A computation-centric distributed graph processing system.
\newblock In {\em 12th $\{$USENIX$\}$ Symposium on Operating Systems Design and
  Implementation ($\{$OSDI$\}$ 16)}, pages 301--316, 2016.

\end{thebibliography}
